\documentclass[sigconf, nonacm, pdfa]{acmart}
\usepackage[a-2b]{pdfx}
\usepackage{hyperref}
\usepackage{amsmath}
\usepackage{enumitem}
\usepackage[update,prepend]{epstopdf}
\usepackage{bbding}
\usepackage{wasysym}
\usepackage{tabularx}
\usepackage{multirow}
\usepackage{ifsym} 
\usepackage{color}
\usepackage{float}
\usepackage{threeparttable}
\usepackage{subfig}
\usepackage{stfloats}
\usepackage{pdfpages}
\usepackage{graphicx}
\usepackage[inkscapelatex=false]{svg}
\usepackage[ruled,linesnumbered, vlined]{algorithm2e}
\usepackage{makecell}
\usepackage[figuresright]{rotating}
\usepackage{eso-pic}
\usepackage{xcolor}
\usepackage[linesnumbered,ruled,vlined]{algorithm2e}
\SetKwInput{KwParam}{Param}

\newcommand{\myparagraph}[1]{\vspace{1ex}\noindent\underline{\it #1.}\xspace}

\newif\ifproofread
\newif\ifextendedreport

\newcommand{\changemarker}[1]{%
\ifproofread
\textcolor{blue}{#1}%
\else
#1%
\fi
}

\newcommand\vldbdoi{10.14778/3705829.3705855}
\newcommand\vldbpages{413 - 425}
\newcommand\vldbvolume{18}
\newcommand\vldbissue{2}
\newcommand\vldbyear{2024}
\newcommand\vldbauthors{\authors}
\newcommand\vldbtitle{\shorttitle} 
\newcommand\vldbavailabilityurl{https://github.com/teijyogen/privbench}
\newcommand\vldbpagestyle{empty} 

\usepackage{amsmath,amsfonts}
\usepackage{weiwMath}
\makeatletter

\newcommand{\Rmnum}[1]{\expandafter\@slowromancap\romannumeral #1@}
\makeatother
\usepackage{graphicx}
\usepackage{textcomp}
\usepackage{xcolor}
\usepackage{float}
\usepackage{url}
\usepackage{pifont}

\newtheorem{theorem}{Theorem}
\newtheorem{definition}{Definition}

\newtheorem{corollary}{Corollary}
\newtheorem{lemma}{Lemma}
\newtheorem{problem}{Problem}

\usepackage{amsmath}

\newcommand{\eat}[1]{}
\newcommand{\kw}[1]{{\ensuremath {\mathsf{#1}}}\xspace}

\newcommand{\ours}{\texttt{PrivBench}\xspace}

\AddToShipoutPictureBG*{%
  \AtPageUpperLeft{%
    \raisebox{-1.5cm}[0pt][0pt]{%
      \makebox[\paperwidth]{%
        \begin{minipage}{0.9\paperwidth}
          \centering
          \color{red}%
          \Large
          This paper has been published in PVLDB and will appear at VLDB 2025. Please cite the following reference:\\
          \color{blue}
          \normalsize
          Yunqing Ge, Jianbin Qin, Shuyuan Zheng, Yongrui Zhong, Bo Tang,
Yu-Xuan Qiu, Rui Mao, Ye Yuan, Makoto Onizuka, and Chuan Xiao.
Privacy-Enhanced Database Synthesis for Benchmark Publishing. PVLDB,
18(2): 413 - 425, 2024. doi:10.14778/3705829.3705855
        \end{minipage}
      }%
    }%
  }%
}

\begin{document}
\proofreadfalse
\extendedreportfalse

\title{Privacy-Enhanced Database Synthesis for Benchmark Publishing}



\author{Yunqing Ge}
\affiliation{%
  \institution{\normalsize{Shenzhen University}}
}
\email{geyunqing2022@email.szu.edu.cn}

\author{Jianbin Qin}
\authornote{Corresponding authors.}
\affiliation{%
  \institution{\normalsize{SICS, Shenzhen University}}
}
\email{qinjianbin@szu.edu.cn}

\author{Shuyuan Zheng}
\authornotemark[1]
\affiliation{%
  \institution{\normalsize{Osaka Univeristy}}
}
\email{zheng@ist.osaka-u.ac.jp}

\author{Yongrui Zhong}
\affiliation{%
  \institution{\normalsize{Shenzhen University}}
}
\email{zhongyongrui2021@email.szu.edu.cn}

\author{Bo Tang}
\affiliation{%
  \institution{\normalsize{Southern University of Science and Technology}}
}
\email{tangb3@sustech.edu.cn}

\author{Yu-Xuan Qiu}
\affiliation{%
  \institution{\normalsize{Beijing Institute of Technology}}
}
\email{qiuyx.cs@gmail.com}

\author{Rui Mao}
\affiliation{%
  \institution{\normalsize{SICS, Shenzhen University}}
}
\email{mao@szu.edu.cn}

\author{Ye Yuan}
\affiliation{%
  \institution{\normalsize{Beijing Institute of Technology}}
}
\email{yuan-ye@bit.edu.cn}

\author{Makoto Onizuka}
\affiliation{%
  \institution{\normalsize{Osaka Univeristy}}
}
\email{onizuka@ist.osaka-u.ac.jp}

\author{Chuan Xiao}
\affiliation{%
  \institution{\normalsize{Osaka Univeristy, Nagoya University}}
}
\email{chuanx@ist.osaka-u.ac.jp}

\begin{abstract}

    Benchmarking is crucial for evaluating a DBMS, yet existing benchmarks often fail to reflect the varied nature of user workloads. 
    As a result, there is increasing momentum toward creating databases that incorporate real-world user data to more accurately mirror business environments. 
    However, privacy concerns deter users from directly sharing their data, underscoring the importance of creating synthesized databases for benchmarking that also prioritize privacy protection. 
    Differential privacy (DP)-based data synthesis has become a key method for safeguarding privacy when sharing data, but the focus has largely been on minimizing errors in aggregate queries or downstream ML tasks, with less attention given to benchmarking factors like query runtime performance. 
    This paper delves into differentially private database synthesis specifically for benchmark publishing scenarios, aiming to produce a synthetic database whose benchmarking factors closely resemble those of the original data. 
    Introducing \textit{PrivBench}, an innovative synthesis framework based on sum-product networks (SPNs), we support the synthesis of high-quality benchmark databases that maintain fidelity in both data distribution and query runtime performance while preserving privacy. 
    We validate that PrivBench can ensure database-level DP even when generating multi-relation databases with complex reference relationships.
    Our extensive experiments show that PrivBench efficiently synthesizes data that maintains privacy and excels in both data distribution similarity and query runtime similarity.
\end{abstract}

\maketitle

\pagestyle{\vldbpagestyle}
\begingroup\small\noindent\raggedright\textbf{PVLDB Reference Format:}\\
\vldbauthors. \vldbtitle. PVLDB, \vldbvolume(\vldbissue): \vldbpages, \vldbyear.\\
\href{https://doi.org/\vldbdoi}{doi:\vldbdoi}
\endgroup
\begingroup
\renewcommand\thefootnote{}\footnote{\noindent
This work is licensed under the Creative Commons BY-NC-ND 4.0 International License. Visit \url{https://creativecommons.org/licenses/by-nc-nd/4.0/} to view a copy of this license. For any use beyond those covered by this license, obtain permission by emailing \href{mailto:info@vldb.org}{info@vldb.org}. Copyright is held by the owner/author(s). Publication rights licensed to the VLDB Endowment. \\
\raggedright Proceedings of the VLDB Endowment, Vol. \vldbvolume, No. \vldbissue\ %
ISSN 2150-8097. \\
\href{https://doi.org/\vldbdoi}{doi:\vldbdoi} \\
}\addtocounter{footnote}{-1}\endgroup

\ifdefempty{\vldbavailabilityurl}{}{
\vspace{.3cm}
\begingroup\small\noindent\raggedright\textbf{PVLDB Artifact Availability:}\\
The source code, data, and/or other artifacts have been made available at \url{\vldbavailabilityurl}.
\endgroup
}

\section{Introduction}
\label{sec:intro}
Benchmarking, a crucial component for evaluating the performance of DBMSs, has historically leveraged established benchmarks such as the TPC series~\cite{TPC}. These conventional benchmarks use fixed schemas and queries to compare the performance of different database systems standardizedly. However, they may fall short in representing the varied, specific workloads and data characteristics unique to every user. Moreover, the nuances of real-world applications, intrinsic data characteristics, and user-specific performance expectations might not be fully captured by these fixed benchmarks, showing a need for benchmarks more tailored to individual users.

For benchmark publishing, synthesizing a database is a highly challenging endeavor because it requires attention to four concerns: (1) fidelity of data distribution, (2) fidelity of query runtime performance, (3) privacy protection, and (4) synthesis efficiency. 
The first two concerns ensure accurate benchmarking, privacy protection accommodates compliance with data protection regulations, and synthesis efficiency facilitates rapid updates of enterprise-level benchmarks.
Some database synthesis efforts, such as SAM~\cite{yang2022sam}, can produce high-fidelity benchmarks but overlook privacy protection.
Such oversight can lead to inapplicability in contexts where data privacy is paramount, e.g., those involving user data.
On the other hand, existing privacy-preserving data synthesis methods (e.g., ~\cite{ping2017datasynthesizer, zhang2017privbayes, zhang2021privsyn, mckenna2021winning, cai2021privmrf, mckenna2022aim, pujol2022prefair, liu2023tabular, cai2023privlava}) only focus on enhancing the fidelity in data distribution while overlooking the fidelity in query runtime performance, thereby far from meeting the practical demands of benchmarking scenarios.

In this paper, we propose PrivBench, a database synthesis framework for benchmark publishing that addresses all four of the aforementioned concerns.
The design of PrivBench is based on our novel observation: 
\textit{sum-product networks}~\cite{poon2011sum}, when used as data synthesis models, can simultaneously ensure the fidelity of data distribution and query runtime while offering linear-time efficiency for data sampling.
However, SPNs do not provide privacy protection. 
Therefore, PrivBench combines SPNs with differential privacy (DP)~\cite{dwork2006differential}, a widely-adopted privacy technique, to construct privacy-preserving SPNs.
Using DP-based SPNs, PrivBench efficiently synthesizes high-fidelity databases, protecting user privacy and reducing the compromise on quality in benchmarking scenarios. 
In particular, PrivBench constructs SPNs on the input database instance, which is equivalent to partitioning the records of each relational table into disjoint blocks and building histograms within each block. Additionally, we connect the per-table SPNs based on reference relationships to form a multi-SPN model. We inject noise into all partitioning operations, histograms, and SPN connections to guarantee DP for data synthesis through the multi-SPN.

\begin{figure*}
        \vspace{-0.3cm}
    \centering
    \includegraphics[width=0.9\linewidth]{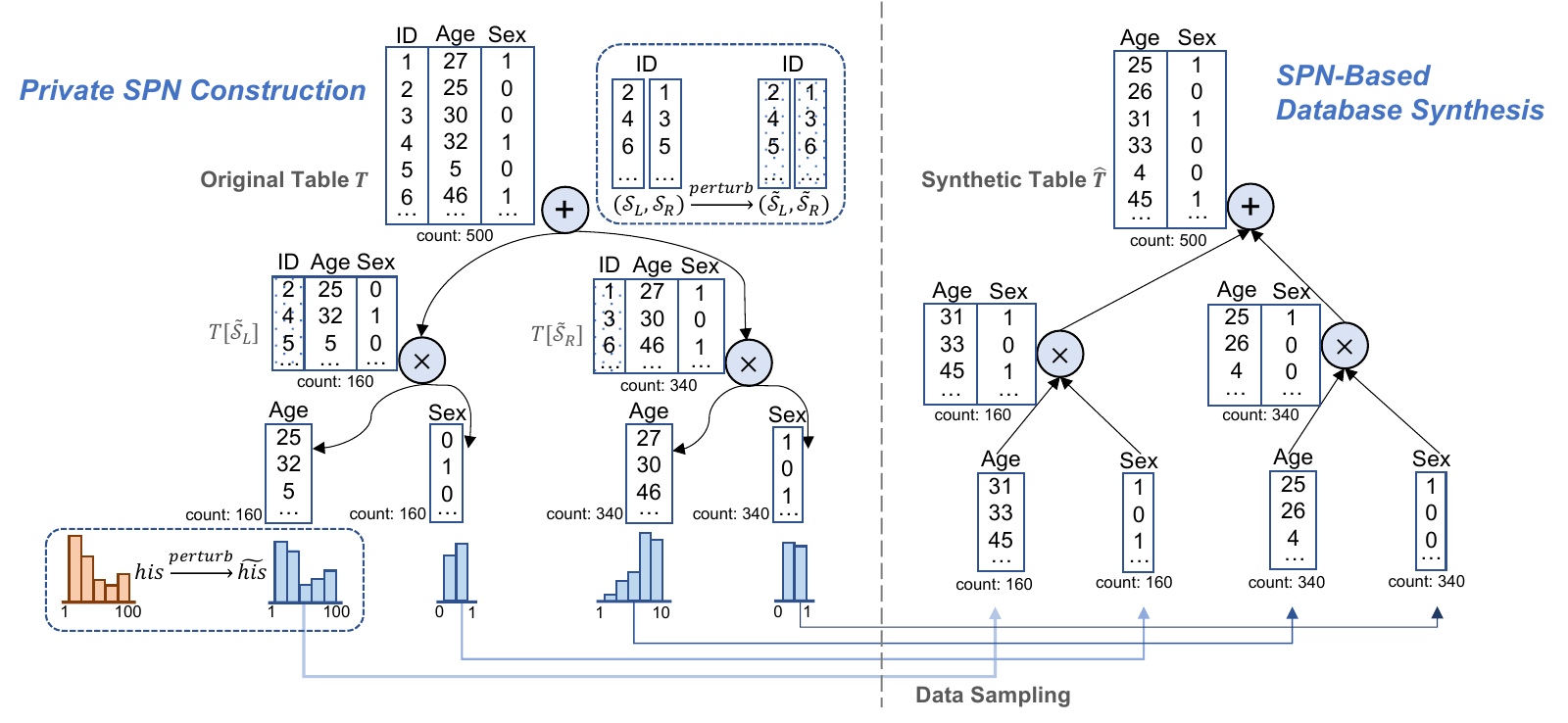} 
            \vspace{-0.3cm}
    \caption{Learning a single-relation SPN and generating data based on the SPN.}
    \label{fig:BuSPNGenData}
                \vspace{-0.2cm}
\end{figure*}

\begin{example}
    Figure \ref{fig:BuSPNGenData} illustrates the process of constructing an SPN for a single table and synthesizing data from it. 
    Initially, the original table comprises two attributes and 500 rows. 
    The data is partitioned into two groups of rows through clustering, which introduces a \textit{sum node}. 
    Subsequently, the columns within each group are separated by a \textit{product node}, with each column forming a leaf node. 
    Each leaf node stores a histogram representing its corresponding single-column subtable. 
    To ensure DP, the table partition decision $(\mathcal{S}_L, \mathcal{S}_R)$ for the sum node, which consists of two sets of row indices or IDs, and the histogram $H$ for each leaf node are perturbed.
    During the database synthesis phase, synthetic data are sampled from these perturbed histograms and progressively merged up the hierarchy. 
    A product node performs a horizontal concatenation of columns, while a sum node executes a vertical concatenation of rows. 
    This recursive procedure is repeated until it reaches the root node, resulting in the generation of a complete synthetic database.
\end{example}

\begin{table}[t]
\footnotesize
\centering
\caption{Comparison with existing differentially private data synthesis methods.}
\setlength{\tabcolsep}{4pt}
\centering
\begin{tabular}{l|l|cc|c}
\hline
\multicolumn{1}{c|}{\multirow{2}{*}{\textbf{Method}}} & \multicolumn{1}{c|}{\multirow{2}{*}{\textbf{\begin{tabular}[c]{@{}c@{}}Multi-\\ relation\end{tabular}}}} & \multicolumn{2}{c|}{\textbf{Similarity}} & \multirow{2}{*}{\textbf{Efficiency}} \\ \cline{3-4}
\multicolumn{1}{c|}{} & \multicolumn{1}{c|}{} & \textbf{Data distr.} & \textbf{Query runtime} &  \\ \hline
AIM~\cite{mckenna2022aim} &  & low & low & low \\
PrivSyn~\cite{zhang2021privsyn} &  & low & med. & med. \\
DataSynthesizer~\cite{ping2017datasynthesizer} &  & low & med. & high \\
PrivBayes~\cite{zhang2017privbayes} &  & low & med. & high \\
Exponential-PreFair~\cite{pujol2022prefair} &  & med. & low & low \\
MST~\cite{mckenna2021winning} &  & med. & low & med. \\
DPGAN~\cite{liu2023tabular} &  & med. & med. & low \\
PrivMRF~\cite{cai2021privmrf} &  & med. & high & low \\
Greedy-PreFair~\cite{pujol2022prefair} &  & med. & high & med. \\
PrivLava~\cite{cai2023privlava} & \multicolumn{1}{c|}{$\checkmark$} & high & med. & low \\
\textbf{PrivBench} & \multicolumn{1}{c|}{\textbf{$\checkmark$}} & \textbf{high} & \textbf{high} & \textbf{high} \\ \hline
\end{tabular}
\label{tab:comparison-alternative}
\end{table}

In essence, our contribution is a novel exploration into the balance between synthesizing enterprise-level databases that maintain a high fidelity to original user data for DBMS benchmarking, while preserving user privacy. 
Through the strategic incorporation of SPNs and DP, PrivBench aligns closely with benchmarking scenarios, protecting user privacy and reducing the compromise on the quality or authenticity of the synthesized database. Table~\ref{tab:comparison-alternative} provides a summary of comparisons with alternative methods.

Our experimental results on several publicly accessible datasets demonstrate that, compared to PrivLava~\cite{cai2023privlava}, the current state-of-the-art (SOTA) in privacy-preserving data synthesis, PrivBench consistently achieves superior performance in both data distribution similarity and query runtime similarity, while significantly enhancing synthesis efficiency.
Notably, PrivBench reduces KL divergence (KLD) by up to $62.4\%$ and Q-error by up to $97.8\%$ relative to PrivLava, which respectively measure the differences in data distribution and query runtime performance between the original and synthetic data.
Our contributions are summarized as follows.

\begin{itemize}[leftmargin=*]
    \item We study the problem of privacy-preserving database synthesis, which transforms a database in a differentially private manner while aligning the performance of query workloads on the synthetic database with the original one. 
    \item To solve the problem, we propose PrivBench, which synthesizes databases by leveraging differentially private SPNs.
    These SPNs not only address the four key concerns in benchmark publishing scenarios but also support multi-relation database synthesis by constructing differentially private fanout tables for primary-foreign key references, adeptly managing complex dependencies among the relations.
    To the best of our knowledge, our work is the first to facilitate differentially private SPNs and the first to apply them to database synthesis.
    \item We conduct a rigorous analysis of privacy and time complexity for PrivBench.
    Our analysis shows how PrivBench ensures that its data synthesis, through the complex multi-SPN construction process, achieves database-level DP~\cite{cai2023privlava}, and that the synthesis time is polynomial in the number of records and the number of attributes in each table.
    \item Our extensive experimental results illustrate that PrivBench consistently achieves better fidelity in both data distribution and query runtime performance while being significantly more efficient than alternative methods.
\end{itemize}

\section{Problem Definition}
\label{sec:problem}

\begin{figure}
        \footnotesize
        \centering
        \subfloat[Primary private table.]{
            \begin{tabular}{ccc}
            \hline
            H-ID            & Rooms                & \dots \\ \hline
            1                    & 2                    & \dots \\
            2                    & 5                    & \dots \\
            3                    & 3                    & \dots \\
            \dots & \dots & \dots \\ \hline
            Range                & {[}1, 10{]}          & \dots \\ \hline
            \end{tabular}
        \label{tab:example_tab_T_1}
        }\subfloat[Secondary private table.]{
                    \begin{tabular}{ccccc}
                    \hline
                    ID           & Sex                  & Age                  & \dots & H-ID            \\ \hline
                    1                    & 1                    & 27                   & \dots & 1                    \\
                    2                    & 0                    & 25                   & \dots & 1                    \\
                    3                    & 0                    & 30                   & \dots & 2                    \\
                    4                    & 1                    & 32                   & \dots & 2                    \\
                    5                    & 0                    & 5                    & \dots & 2                    \\
                    6                    & 1                    & 46                   & \dots & 3                    \\
                    \dots & \dots & \dots & \dots & \dots \\ \hline
                    Range                & {[}0, 1{]}           & {[}1, 100{]}         & \dots &                      \\ \hline
                    \end{tabular}
        \label{tab:example_tab_T_2}   
        }
    \caption{A database consisting of two private tables.}
    \label{tab:example_tab}
\end{figure}

\changemarker{
We consider a scenario where a benchmark publisher, such as an enterprise or organization, synthesizes a database benchmark using private data (e.g., user data) and releases it to other parties, referred to as benchmark consumers, for testing purposes. The benchmark publisher is a trusted party who is granted access to the original database by the data entities (e.g., users).
In contrast, the benchmark consumers may be malicious and may try to infer private information from the released benchmark.
Therefore, data anonymization techniques, such as differential privacy~\cite{dwork2014algorithmic}, are required to safeguard data privacy and ensure compliance with data protection regulations like the GDPR.
}

Consider a benchmark that consists of a database instance $D = \set{T_1, T_2,\dots, T_n}$ with its associated schema, each $T_i \in D$ being a private relation table. 
\changemarker{
For any $T_i, T_j \in D$, a record $r_i \in T_i$ refers to a record $r_j \in T_j$, denoted as $r_i \to r_j$, if the foreign key of $r_i$ refers to the primary key of $r_j$ and the values of the two keys match.
A record $r_i$ depends on $r_j$, denoted as $r_i \rightsquigarrow r_j$, if either $r_i \to r_j$ or $r_i$ refers to another record $r_k \in T_k$ ($k\neq i, j$) that depends on $r_j$. 
We assume $r \rightsquigarrow r$ never occurs for any record $r$. 
}

Without loss of generality, following previous work \cite{cai2023privlava}, we assume that $T_1$ is a \textit{primary} private table (e.g., the table in Figure \ref{tab:example_tab_T_1}) that contains sensitive information, and any other table $T_i, i\in[2, n]$, is a \textit{secondary} private table (e.g., the table in Figure \ref{tab:example_tab_T_2}) in which the records depend on the records in $T_1$.
We study the case of synthesizing a database $\widehat{D}$ similar to $D$ while preserving the privacy of all private tables. 
We employ DP~\cite{dwork2014algorithmic}, a celebrated approach for protecting database privacy. 
\begin{definition}[Table-Level DP~\cite{dwork2014algorithmic}]
\label{def:table-dp}
    A randomized algorithm $\mathcal{M}$ that takes as input a table $T$ satisfies table-level $\epsilon$-DP ($\epsilon \geq 0$) if for any possible output $o$, and for any pair of neighboring tables $T, T'$ that differ in the value of only one record, we have $$\Pr{\mathcal{M}(T) = o} \leq \exp(\epsilon) \cdot \Pr{\mathcal{M}(T') = o}.$$
\end{definition}
\begin{definition}[Database-Level DP \cite{cai2023privlava}]
\label{def:db-dp}
    \changemarker{A randomized algorithm $\mathcal{M}$ that takes as input a database $D$ satisfies database-level $\epsilon$-DP ($\epsilon \geq 0$) if for any possible output $o$ and for any pair of neighboring databases $D, D'$, we have $$\Pr{\mathcal{M}(D) = o} \leq \exp(\epsilon) \cdot \Pr{\mathcal{M}(D') = o}.$$
    $D$ and $D'$ neighbor if for the primary table $T_1$, they differ in the value of only one record $r_1 \in T_1$, and for the other tables $T_i, i\in[2, n]$, they differ in the values of all records $r_i \rightsquigarrow r_1$ that depend on $r_1$.}
\end{definition}

\textit{Table-level} DP provides a strong guarantee that the output distributions remain indistinguishable before and after any change to a single record $r\in T$.
This guarantee should hold even when accounting for dependencies among records induced by foreign keys in the relational schema.
For example, in Figure \ref{tab:example_tab}, if we aim to protect a household in the primary table with multiple household members in the secondary table, we need to ensure that any holistic changes to the information of each household and its members do not significantly affect the output.
Therefore, we adopt \textit{database-level} DP \cite{cai2023privlava}, which further ensures that any changes in other records resulting from modifying a record $r_1 \in T_1$ also have a limited impact on the output.
To ensure a bounded multiplicity for neighboring databases, we follow prior work \cite{kotsogiannis2019privatesql, tao2020computing, dong2022r2t, cai2023privlava} to assume that each record in $T_1$ is referred to by at most $\tau_i$ records for each secondary table $T_i$.
The privacy budget $\epsilon$ controls the level of indistinguishability.
The smaller the privacy budget, the greater the level of privacy protection.
With the definition of DP, we target the following privacy-preserving database synthesis problem.
\begin{problem}[Privacy-Preserving Database Synthesis]
  Given a database $D$ with its schema, the task is to generate a database $\widehat{D}$ with the same schema as $D$ while ensuring database-level $\epsilon$-DP. 
\end{problem}



Whereas many differentially private data synthesizers are available~\cite{ping2017datasynthesizer, zhang2017privbayes, zhang2021privsyn, mckenna2021winning, cai2021privmrf, mckenna2022aim, pujol2022prefair, liu2023tabular, cai2023privlava}, their aims are mainly reducing the errors of aggregate queries or optimizing downstream ML tasks, rather than publishing a database for benchmarking. 
\changemarker{
Considering that the benchmark is used to test the performance of a DBMS, we aim to synthesize a database $\widehat{D}$ with the following concerns.
}

\changemarker{
\myparagraph{1. Data distribution similarity}
$\widehat{D}$ should be statistically similar to $D$.
Given a table $T \in D$ and its counterpart $\widehat{T} \in \widehat{D}$, we measure their statistical similarity by the KL divergence (KLD) of $T$ from $\widehat{T}$: 
\begin{align*}
    \kw{KLD}(T \parallel \widehat{T}) = \sum\nolimits_{x \in \mathcal{X}(T\cup \widehat{T})} \Pr{x | T} \log \left( \frac{\Pr{x | T}}{\Pr{x | \widehat{T}}} \right)
\end{align*}
where $\Pr{x | T}$ denotes the probability of row value $x$ in $T$, and $\mathcal{X}(T)$ denotes the set of row values in $T$. Then, the KLD of $D$ from $\widehat{D}$ is defined as the average KLD over all the tables of $D$. 
}

\changemarker{
\myparagraph{2. Query runtime similarity}
For the query workloads to be executed on $D$, we expect they report the same runtime performance on $\widehat{D}$. 
Given a query, we can compare the execution times on $D$ and $\widehat{D}$. Since the execution time may vary across DBMSs, we can also compare the cardinality, i.e., the number of records in the query result, which is often used in a query optimizer for estimating query performance~\cite{leis2015good}. In particular, Q-error~\cite{moerkotte2009preventing} is a widely used measure for comparing cardinalities: 
\begin{align*}
    \kw{QE}(q, D, \widehat{D}) = \max \left( \frac{Card(q, D)}{Card(q, \widehat{D})}, \frac{Card(q, \widehat{D})}{Card(q, D)} \right) 
\end{align*}
where $q$ denotes a query in the workload, and $Card(q, D)$ denotes the cardinality of executing $q$ on $D$. Then, we compute the mean Q-error of the queries in the workload. 
}

\changemarker{
\myparagraph{3. Synthesis efficiency}
$\widehat{D}$ should be synthesized efficiently.
In practice, databases often contain a large number of records and attributes. 
As these databases are frequently updated with daily activities, the synthetic benchmark also needs to be updated accordingly, resulting in significant time costs.
Therefore, we should ensure that the synthesis algorithm completes in polynomial time w.r.t. both the number of records and the number of attributes.
}

\section{PrivBench} 
\label{sec:method}
In this section, we propose PrivBench, an SPN-based differentially private database synthesis method.

\subsection{Sum-Product Network}

\changemarker{
An SPN \cite{poon2011sum} is a rooted acyclic-directed graph that represents the data distribution of a dataset.
PrivBench utilizes SPNs to address the three concerns outlined in Section~\ref{sec:problem} due to the following merits.
\begin{enumerate}[leftmargin=*]
    \item SPNs have a strong ability to capture data distributions.
    Current SOTAs~\cite{cai2021privmrf, cai2023privlava} for privacy-preserving data synthesis have employed graphical models (GMs) other than SPNs to learn data distributions.
    However, it has been proven that all tractable GMs can be transformed into equivalent SPNs and that SPNs are even strictly more general, implying that SPNs can learn more accurate data distributions than other GMs \cite{poon2011sum}.
    \item SPNs have demonstrated excellent performance in cardinality estimation~\cite{DeepDB, han2021cardinality}.
    Our insight is that, since computing the cardinalities of queries on synthetic data can be approximately viewed as estimating the cardinalities on the original data, we believe that SPNs should perform well in optimizing the Q-error metric, resulting in better query runtime similarity. 
    \item Inference in SPNs finishes in time linear to the number of nodes \cite{poon2011sum}, making sampling data from SPNs notably efficient.
\end{enumerate}
In a nutshell, the use of SPNs for database synthesis aligns with the major concerns of benchmark publishing except that SPNs do not provide privacy protection.
}

\changemarker{
Following prior work on cardinality estimation~\cite{DeepDB}, we employ binary tree-structured SPNs with sum and product nodes as internal nodes and leaves.
Formally, given a table $T$ with attributes $X_1,\dots,X_M$, a sum (resp. product) node $v_\kw{sum}$ (resp. $v_\kw{prod}$) splits the rows (resp. columns) into two subtables $T[\mathcal{S}_L], T[\mathcal{S}_R]$, where $\mathcal{S}_L, \mathcal{S}_R$ are two subsets of row (resp. column) indices.
A leaf node $v_\kw{leaf}$ represents the marginal distribution of an attribute $X_m, m\in[M]$ w.r.t. a subset $\mathcal{S}$ of rows using a histogram, i.e., $v_\kw{leaf} =  \Pr{X_m | \mathcal{S}}$.
Then, the value of a sum node is the weighted sum of its children $v_L, v_R$, i.e., $v_\kw{sum} =  w_L \cdot v_L + w_R \cdot v_R$, where $w_L = \frac{|\mathcal{S}_L|}{|\mathcal{S}_L| + |\mathcal{S}_R|}, w_R = \frac{|\mathcal{S}_R|}{|\mathcal{S}_L| + |\mathcal{S}_R|}$ are weights proportional to the sizes of the split clusters,
while the value of a product node is the product of its children, i.e.,  $v_\kw{prod} =  v_L \cdot v_R$.
Consequently, the value of the root node approximates the joint distribution $\Pr{X_1,...,X_M}$ for a given table by summing and multiplying the marginal distributions represented by the leaves from bottom to top.
}

\SetKwInOut{Input}{Input}
\SetKwInOut{Output}{Output}

\begin{algorithm}[t]
\small
    \caption{PrivBench}
    \label{alg:privbench}
    \Input{Database $D$}
    \Output{Synthetic database $\widehat{D}$}
    \KwParam{Privacy budgets $\epsilon^s_i$ (for SPN construction) and $\epsilon^f_i$ (for fanout construction), $\forall i \in [n]$}
    \For{each private table $T_i$ of $D$ }{
        $t_i \leftarrow \kw{PrivSPN}(T_i, \epsilon^s_i)$\;
    }
    \For{each pair of tables $(T_i, T_j)$ where $T_i$ refers to $T_j$}{
        $t_i' \leftarrow \kw{PrivFanout}(T_i, t_i, FK_{i,j}, \epsilon^f_i)$\;
    }
    \For{each modified SPN $t_i'$}{
        \small{$\widehat{T}_i \leftarrow \kw{SampleDataFromSPN}(t_i')$\;}
    }
    \Return $\widehat{D}=\set{\widehat{T}_1,\dots,\widehat{T}_n}$\;
\end{algorithm}

\subsection{Overview of PrivBench}
As shown in Algorithm \ref{alg:privbench}, PrivBench involves a three-phase process for database synthesis.
\begin{itemize}[leftmargin=*]
    \item \textit{Private SPN Construction (Lines 1--2).} 
    Firstly, we construct an SPN $t_i$ for each table $T_i$ in the input database $D$. 
    Each SPN $t_i$ is created by a differentially private algorithm $\kw{PrivSPN}$, where each leaf represents the marginal of an attribute for some rows.
    \item \textit{Private Fanout Construction (Lines 3--4).}
    Secondly, we complement each SPN $t_i$ with some leaf nodes that model primary-foreign key references to obtain a modified SPN $t_i'$.
    We use a differentially private algorithm $\kw{PrivFanout}$ to calculate \textit{fanout frequencies} for each foreign key and store them in the complemented leaf nodes.  
    \item \textit{SPN-Based Database Synthesis (Lines 5--6)} 
    Thirdly, we sample a table $\widehat{T}_i$ from each modified SPN $t_i'$ to synthesize a database $\widehat{D}$.    
\end{itemize}

\subsection{Private SPN Construction}

\begin{algorithm} [t]
\small
  \caption{Private SPN Construction $\kw{PrivSPN}(T, \epsilon)$}
  \label{alg:SPN}
  \SetKwComment{tcp}{//}{}
  \Input{table $T$, total privacy budget $\epsilon$. } 
  \Output{A tree of SPN $t=(\kw{parent}, \kw{left}, \kw{right})$}
  \SetKwProg{Procedure}{procedure}{}{}
  $\kw{op}, \epsilon_\kw{op}, \bar{\epsilon} \leftarrow \kw{Planning}(T, \epsilon)$\;
    $\kw{parent}, (\widetilde{\mathcal{S}}_L, \widetilde{\mathcal{S}}_R) \leftarrow \kw{ParentGen}(T, \kw{op}, \epsilon_\kw{op})$\;
    $\kw{left}, \kw{right} \leftarrow \kw{ChildrenGen}(T, \kw{op}, \widetilde{\mathcal{S}}_L, \widetilde{\mathcal{S}}_R, \bar{\epsilon})$\;
    \Return{$(\kw{parent}, \kw{left}, \kw{right})$}\;
    \Procedure{$\kw{ParentGen}(T, \kw{op}, \epsilon_\kw{op})$}{
        \If{$\kw{op} = \kw{OP.LEAF}$}{
            \small{$\widetilde{\mathcal{S}}_L, \widetilde{\mathcal{S}}_R \leftarrow \emptyset, \emptyset$\;
            $\widetilde{his} \leftarrow \kw{his}(T) + \kw{Lap}(\frac{\Delta(\kw{his})}{\epsilon_\kw{op}})$; $\kw{parent} \leftarrow \kw{LeafNode}(\widetilde{his})$\;}
            
        }    
        \ElseIf{$\kw{op} = \kw{OP.SUM}$}{
            \footnotesize{$\widetilde{\mathcal{S}}_L, \widetilde{\mathcal{S}}_R \leftarrow \kw{RowSplit}(T, \epsilon_\kw{op})$; $\kw{parent} \leftarrow \kw{SumNode}(\widetilde{\mathcal{S}}_L, \widetilde{\mathcal{S}}_R)$\;}
        } 
        \ElseIf{$\kw{op} = \kw{OP.PRODUCT}$}{
            \footnotesize{$\widetilde{\mathcal{S}}_L, \widetilde{\mathcal{S}}_R \leftarrow \kw{ColSplit}(T, \epsilon_\kw{op})$; $\kw{parent} \leftarrow \kw{ProdNode}(\widetilde{\mathcal{S}}_L, \widetilde{\mathcal{S}}_R)$\;}
        } 
        \Return{$\kw{parent}, (\widetilde{\mathcal{S}}_L, \widetilde{\mathcal{S}}_R)$}\;
    }
    \Procedure{$\kw{ChildrenGen}(T, \kw{op}, \widetilde{\mathcal{S}}_L, \widetilde{\mathcal{S}}_R, \bar{\epsilon})$}{
        \If{$\kw{op} = \kw{OP.SUM}$}{
            $\epsilon_L \leftarrow \bar{\epsilon}, \epsilon_R \leftarrow \bar{\epsilon}$
        } 
        \ElseIf{$\kw{op} = \kw{OP.PRODUCT}$}{
            $\epsilon_L \leftarrow \bar{\epsilon} \cdot \frac{\sigma(T[\widetilde{\mathcal{S}}_L])}{\sigma(T[\widetilde{\mathcal{S}}_L]) + \sigma(T[\widetilde{\mathcal{S}}_R])}, \epsilon_R \leftarrow \bar{\epsilon} - \epsilon_L$
        }
        \If{$\kw{op} = \kw{OP.LEAF}$}{
            $\kw{left} \leftarrow \kw{null}$; $\kw{right} \leftarrow \kw{null}$\;
        }
        \Else{
            \footnotesize{$\kw{left} \leftarrow \kw{PrivSPN}(T[\widetilde{\mathcal{S}}_L], \epsilon_{L})$; $\kw{right} \leftarrow \kw{PrivSPN}(T[\widetilde{\mathcal{S}}_R], \epsilon_R)$\;}
        }
        \Return{$\kw{left}, \kw{right}$}
    }
\end{algorithm}

\subsubsection{Overview of \kw{PrivSPN}}
Algorithm \ref{alg:SPN} outlines the process for constructing an SPN on a single table, \changemarker{which calls Algorithm \ref{alg:plan} to optimize the SPN's structure}.
The \kw{PrivSPN} algorithm recursively generates a binary tree $t=(\kw{parent}, \kw{left}, \kw{right})$, where \kw{parent} is a node, and \kw{left}, \kw{right} are the two child subtrees.
We divide \kw{PrivSPN} into the following three procedures.
\begin{itemize}[leftmargin=*]
    \item \textit{\kw{Planning} (Operation Planning):} 
    We decide an operation \kw{op} for the parent node and the privacy budget $\epsilon_\kw{op}$ allocated to \kw{op}.
    \item \textit{\kw{ParentGen} (Parent Generation):} 
    According to the operation \kw{op} and the privacy budget $\epsilon_\kw{op}$, we generate the parent node.
    \item \textit{\kw{ChildrenGen} (Children Generation):} We further generate the children \kw{left}, \kw{right} through recursive calls to \kw{PrivSPN}.
\end{itemize}
\changemarker{
For ease of presentation, we first introduce how the parent and children are generated, and then we delve into the \kw{Planning} procedure.}

\begin{algorithm}[t]
\small
\caption{Row Splitting \kw{RowSplit}}
    \label{alg:rowsplit}
    \Input{table $T$, privacy budget $\epsilon$}
    \Output{row partition $(\widetilde{\mathcal{S}}_L, \widetilde{\mathcal{S}}_R)$}
    \KwParam{number of iterations $J$, minimum table size $\beta$}
    \changemarker{
    $(\mathcal{S}_{L, 0},  \mathcal{S}_{R, 0})\leftarrow$ Sample a row partition such that $|\mathcal{S}_{L, 0}| = |T| / 2$\;
    \For{$j$ from $1$ to $J$}{
        $c_L \leftarrow \sum_{r\in T[\mathcal{S}_{L, j-1}]} r / |\mathcal{S}_{L, j-1}|$  \quad// Center of left cluster\;
        $c_R \leftarrow \sum_{r\in T[\mathcal{S}_{R, j-1}]} r / |\mathcal{S}_{R, j-1}|$ \quad// Center of right cluster\;  
        $\kw{diff}_i \leftarrow \kw{dist}(T[i] - c_L) - \kw{dist}(T[i] - c_R), \forall i \in [|T|]$\;
        $\widetilde{\kw{diff}}_i \leftarrow \kw{diff}_i + \kw{Lap}(\frac{\Delta(\kw{diff}_i)\cdot J}{\epsilon}), \forall i \in [|T|]$\;
        $\mathcal{S}_{L, j} \leftarrow \arg_{i} (\widetilde{\kw{diff}}_i \leq 0)$;
        $\mathcal{S}_{R, j} \leftarrow \arg_{i} (\widetilde{\kw{diff}}_i > 0)$\;
    }
    Adjust clusters $\mathcal{S}_{L, J}, \mathcal{S}_{R, J}$ to ensure their sizes are not less than $\beta$\;
    \Return $(\widetilde{\mathcal{S}}_L, \widetilde{\mathcal{S}}_R)= (\mathcal{S}_{L, J}, \mathcal{S}_{R, J})$\;
    }
\end{algorithm}

\begin{algorithm}[ht]
\small
\caption{Column Splitting \kw{ColSplit}}
    \label{alg:colsplit}
    \Input{table $T$, privacy budget $\epsilon$}
    \Output{column partition $(\widetilde{\mathcal{S}}_L, \widetilde{\mathcal{S}}_R)$}
    \changemarker{
    \For{$j$ from $1$ to $|attr(T)|$}{
        \small{$o_i = (\mathcal{S}_{L, j},  \mathcal{S}_{R, j})\leftarrow$ Sample a column partition such that $|\mathcal{S}_{L, j}| = |attr(T)| / 2$\;}
        $\rho_j \leftarrow \kw{NMI}(T, (\mathcal{S}_{L, j},\mathcal{S}_{R, j}))$ 
    }
    $(\widetilde{\mathcal{S}}_L, \widetilde{\mathcal{S}}_R) \leftarrow$ Sample a partition from the set of partitions $O = \{o_j\}_{\forall j}$ with probability defined in Equation (\ref{eq:exponential}) \;
    \Return $(\widetilde{\mathcal{S}}_L, \widetilde{\mathcal{S}}_R)$\;
    }
\end{algorithm}

\subsubsection{Parent generation}
We generate a parent node based on the operation $\kw{op}$ and perturb it with the privacy budget $\epsilon_\kw{op}$, as follows.

\myparagraph{Case 1: $\kw{op} = \kw{OP.LEAF}$  (Line 6)}
We generate a leaf node $\kw{LeafNode}(\widetilde{his})$ with a perturbed histogram $\widetilde{his}$.
Specifically, the $\kw{ParentGen}$ procedure first computes a histogram $\kw{his}(T)$ over the table $T$ and then perturbs the histogram by adding some Laplace noise $\kw{Lap}(\frac{\Delta(\kw{his})}{\epsilon_\kw{op}})$, where $\Delta$ returns the global sensitivity of a given function $f$, i.e., $\Delta(f) = \max_{T, T'} |f(T) - f(T')|$.

\myparagraph{Case 2: $\kw{op} = \kw{OP.SUM}$ (Line 9)} 
We call the \kw{RowSplit} mechanism (see Algorithm \ref{alg:rowsplit}) to compute a differentially private row partition $(\widetilde{\mathcal{S}}_L, \widetilde{\mathcal{S}}_R)$, where $\widetilde{\mathcal{S}}_L, \widetilde{\mathcal{S}}_R$ are two subsets of row indices.
Then, the row partition yields the sum node $\kw{SumNode}(\widetilde{\mathcal{S}}_L, \widetilde{\mathcal{S}}_R)$.

For row splitting, our strategy is to group similar rows into the same cluster. 
By generating distinct segments, SPNs can more effectively learn and represent the local distributions of different data subsets. 
To achieve this, we adopt the DP-based $K$-Means algorithm, DPLloyd~\cite{su2016differentially}, which groups data points into $K$ clusters by comparing their distances to $K$ cluster centers. 
However, the output of DPLloyd is the cluster centers rather than the data points within the clusters, which does not work in our task.
Therefore, we adapt DPLloyd to the \kw{RowSplit} mechanism to support yielding a row partition $(\widetilde{\mathcal{S}}_L, \widetilde{\mathcal{S}}_R)$ as output.

Concretely, as shown in Algorithm \ref{alg:rowsplit}, we first uniformly sample a row partition $(\mathcal{S}_{L, 0},  \mathcal{S}_{R, 0})$ where the clusters are half-sized (Line 1). 
Then, we update the partition through $J$ iterations. 
In each iteration $j$, we first calculate the centers $c_L, c_R$ of the clusters $\mathcal{S}_{L, j-1},  \mathcal{S}_{R, j-1}$ (Lines 3-4). 
Next, for each row $T[i]$, we compute the difference $\kw{diff}_i$ between its distances $\kw{dist}(T[i], c_L), \kw{dist}(T[i], c_R)$ to the two centers (Line 5) and add some Laplace noise $\kw{Lap}(\frac{\Delta(\kw{diff}_i)\cdot J}{\epsilon})$ to ensure DP (Line 6), where $\Delta(\kw{diff}_i)$ equals the absolute difference between the supremum and infimum of the distance function $\kw{dist}_i$. 
Then, if the perturbed difference $\widetilde{\kw{diff}}_i$ is positive, the row $T[i]$ is considered closer to the center $c_R$ and thus should be assigned to the right cluster $\mathcal{S}_{R, j}$; 
otherwise, it is assigned to $\mathcal{S}_{L, j}$ (Line 7).
After $J$ iterations, we adjust the sizes of the clusters $\mathcal{S}_{L, J}, \mathcal{S}_{R, J}$ through random sampling to ensure each of them has at least $\beta$ rows (Line 8);
this adjustment enhances the utility of the perturbed histogram in each leaf node, as a larger histogram is more resilient to DP noise.
Note that the distance function $\kw{diff}_i$ is customizable. 
In our experiments, we use the L1 distance for each numerical attribute and the Hamming distance for each categorical attribute; 
the distance function is defined as the sum of the attribute-wise distances.

\myparagraph{Case 3: $\kw{op} = \kw{OP.PRODUCT}$ (Line 11)}
We generate a product node $\kw{ProdNode}(\widetilde{\mathcal{S}}_L, \widetilde{\mathcal{S}}_R)$ by calling the \kw{ColSplit} mechanism (see Algorithm \ref{alg:colsplit}).
\kw{ColSplit} determines a differentially private column partition $(\widetilde{\mathcal{S}}_L, \widetilde{\mathcal{S}}_R)$, where $\widetilde{\mathcal{S}}_L, \widetilde{\mathcal{S}}_R$ are subsets of column indices.

\changemarker{For column splitting, our goal is to divide the given table into two subtables with low correlation so that the product of the marginal distributions of the subtables approximates the original joint distribution.
Therefore, the \kw{ColSplit} mechanism attempts to minimize the normalized mutual information (NMI, see Definition \ref{def:nmi}) of the split subtables.
As shown in Algorithm \ref{alg:colsplit}, we first construct a candidate set $O$ of column partitions (Lines 1-3) and then sample a partition from the candidates as the output (Lines 4-5).
Note that we conduct $|attr(T)|$ iterations to select $|attr(T)|$ column partitions, rather than all possible column partitions, as the candidate set because it makes the computation tractable.
In each iteration $j$, we uniformly sample a partition $o_j=(\mathcal{S}_{L,j}, \mathcal{S}_{R,j})$ from all half-sized column partitions such that $|\mathcal{S}_{L,j}|= |attr(T)| / 2$ since we observed that subtables with similar dimensionalities generally exhibit lower NMI.
Subsequently, we calculate NMI $\rho_j$ for candidate $o_j$ (Line 3).
After the iterations, we sample a partition $(\widetilde{\mathcal{S}}_L, \widetilde{\mathcal{S}}_R)$ from the candidate set $O$ with the following probability distribution (Line 4):
\begin{equation}
\label{eq:exponential}
    \Pr{(\widetilde{\mathcal{S}}_L, \widetilde{\mathcal{S}}_R) = o_j} = \frac{\kw{exp}(\frac{- \epsilon\cdot \rho_j}{2 \Delta(\kw{NMI} | O)})}{\sum_{j' \in \big[|attr(T)|\big]} \kw{exp}(\frac{-\epsilon\cdot \rho_{j'}}{2 \Delta(\kw{NMI}| O)})},
\end{equation}
where $\Delta(\kw{NMI} | O)=\max_{T,T', o\in O} |\kw{NMI}(T, o) - \kw{NMI}(T', o)|$ is the sensitivity of the NMI w.r.t candidates $O$.
Intuitively, the smaller the correlation between subtables $T[\widetilde{\mathcal{S}}_L], T[\widetilde{\mathcal{S}}_R]$, the lower the information loss from column splitting.
Therefore, the lower the NMI $\rho_j$, the higher the probability of sampling the partition $o_j$ in Equation (\ref{eq:exponential}).
Moreover, as the privacy budget $\epsilon$ decreases, the sampling probabilities become closer to a uniform distribution, increasing the likelihood of outputting a partition with high correlation.
}
\changemarker{
\begin{definition}[Normalized Mututal Information]
\label{def:nmi}
The entropy $H$ of a table $T$ is $H(T) = \sum_{x\in \mathcal{X}(T)} \Pr{x | T} \log_2 \frac{1}{\Pr{x | T}}$, where $\mathcal{X}(T)$ is the set of row values in $T$.
For a row or column partition $(\mathcal{S}_{L}, \mathcal{S}_{R})$ of table $T$, the NMI is:
\begin{equation*}
    \kw{NMI}(T, (\mathcal{S}_{L}, \mathcal{S}_{R})) = \frac{H(T[\mathcal{S}_{L}]) + H(T[\mathcal{S}_{R}]) - H(T)}{\sup(H(T))},
\end{equation*}
where $\sup(H(T))=\log_{2}|T|$ is the upper bound of $H(T)$.
\end{definition}
}

\subsubsection{Children generation}
If $\kw{op} = \kw{OP.LEAF}$, the children are set to \kw{null} because the \kw{parent} node is a leaf node (Lines 19--20).
Otherwise, we generate the \kw{left} and \kw{right} subtrees for the split subtables $T[\widetilde{\mathcal{S}}_L], T[\widetilde{\mathcal{S}}_R]$.
Specifically, we first allocate the remaining privacy budget $\bar{\epsilon}$ among the subtrees (Lines 15--18), which will be discussed in Section \ref{sec:dp:privspn}.
Then, given the allocated privacy budgets $\epsilon_L, \epsilon_R$, we call the \kw{PrivSPN} procedure for the subtables $T[\widetilde{\mathcal{S}}_L], T[\widetilde{\mathcal{S}}_R]$ to generate the \kw{left} and \kw{right} subtrees, respectively (Line 22).

\begin{algorithm}[t]
\small
    \caption{Operation Planning $\kw{Planning}(T, \epsilon)$}
    \label{alg:plan}
    \Input{table $T$, privacy budget $\epsilon$}
    \Output{next operation $\kw{op}$, privacy budget $\epsilon_\kw{op}$ for $\kw{op}$, remaining privacy budget $\bar{\epsilon}$}
    \KwParam{threshold for column splitting $\alpha$, minimum table size $\beta$, budget ratios $\gamma_1, \gamma_2$}
    \SetKwProg{Procedure}{procedure}{}{}
        \If{$|attr(T)| = 1$}{
        $\kw{op} \leftarrow \kw{OP.LEAF}$; $\epsilon_\kw{op} \leftarrow \epsilon$; $\bar{\epsilon}\leftarrow 0$\; 
    }
    \Else{
        $\tilde{\rho}, \epsilon_\kw{eval} \leftarrow \kw{CorrTrial}(T, \epsilon)$\; 
        $\kw{op} \leftarrow \kw{DecideOP}(T, \tilde{\rho})$\;
        $\epsilon_\kw{op}, \bar{\epsilon} \leftarrow \kw{AllocBudgetOP}(T, \kw{op}, \epsilon, \epsilon_\kw{eval})$\; 
    }
    \Return\ $\kw{op}, \epsilon_\kw{op}, \bar{\epsilon}$\;
        \Procedure{$\kw{CorrTrial}(T, \epsilon)$}{
            \If{$|T| \geq 2\beta$ and $|attr(T)| > 1$}{
                $\epsilon_\kw{eval} \leftarrow \epsilon \cdot \gamma_1 / \sigma(T)$\;
                $(\widetilde{\mathcal{S}}_L, \widetilde{\mathcal{S}}_R) \leftarrow \kw{ColSplit}(T, \epsilon_\kw{eval} \cdot \gamma_2)$\;
                $\tilde{\rho} \leftarrow \kw{NMI}(T, (\widetilde{\mathcal{S}}_L, \widetilde{\mathcal{S}}_R)) + \kw{Lap}(\frac{\Delta(\kw{NMI})}{\epsilon_\kw{eval} \cdot (1-\gamma_2)})$\;
            }
            \Else{
                $\epsilon_\kw{eval} \leftarrow 0$; $\tilde{\rho} \leftarrow 0$\;
            }
            \Return\ $\tilde{\rho}, \epsilon_\kw{eval}$
        }
        \Procedure{$\kw{DecideOP}(T, \tilde{\rho})$}{
            \If{$|T| \geq 2\beta$ and $|attr(T)| > 1$}{
                \If{$\tilde{\rho} \leq \alpha$}{
                    $\kw{op} \leftarrow \kw{OP.PRODUCT}$
                }
                \Else{
                    $\kw{op} \leftarrow \kw{OP.SUM}$
                }
            }
            \Else{
                $\kw{op} \leftarrow \kw{OP.PRODUCT}$
            }
            \Return{\kw{op}}
        }
        \Procedure{$\kw{AllocBudgetOP}(T, \kw{op}, \epsilon, \epsilon_\kw{eval})$}{
            \If{\kw{op} = \kw{OP.PRODUCT} and $|attr(T)| = 2$}{
                $\epsilon_\kw{op} \leftarrow 0$\;
            }
            \Else{
                $\epsilon_\kw{op} \leftarrow \epsilon/\sigma(T) - \epsilon_\kw{eval}$\;
            }
       
            \Return{$\epsilon_\kw{op}, \epsilon - \epsilon_\kw{eval} - \epsilon_\kw{op}$}
        }
\end{algorithm}

\subsubsection{Planning}
The \kw{Planning} procedure takes the role of determining an operation $\kw{op}$ for the parent node, which can be creating a leaf/sum/product node.
Since all these types of operations should be performed in a differentially private manner, \kw{Planning} also allocates a privacy budget $\epsilon_\kw{op}$ for the parent node and calculates the remaining privacy budget $\bar{\epsilon}$ for the children.
Concretely, \kw{Planning} proceeds as follows.

\myparagraph{1. Leaf node validation (Lines 1--2)}
We validate if we should generate a leaf node as the parent node.
When the given table has only one attribute, i.e., $|attr(T)| = 1$, the operation must be creating a leaf node, i.e., $\leftarrow \kw{OP.LEAF}$, and all the given privacy budget $\epsilon$ should be allocated for \kw{op}, i.e., $\epsilon_\kw{op} = \epsilon$.

\myparagraph{2. Correlation trial (procedure \kw{CorrTrial})}
We preliminarily test the performance of column splitting, which will guide the selection of the operation \kw{op} through the subsequent \kw{DecideOP} procedure.
Specifically, we consume a privacy budget $\epsilon_\kw{eval}\cdot \gamma_2$ to generate a trial column partition $(\widetilde{\mathcal{S}}_L, \widetilde{\mathcal{S}}_R)$ and an additional privacy budget of $\epsilon_\kw{eval}\cdot(1-\gamma_2)$ to compute its noisy NMI $\tilde{\rho}$ (Lines 11-12):
$$\tilde{\rho}=\kw{NMI}(T, (\widetilde{\mathcal{S}}_L, \widetilde{\mathcal{S}}_R)) + \kw{Lap}(\frac{\Delta(\kw{NMI})}{\epsilon_\kw{eval} \cdot (1-\gamma_2)}).$$
\changemarker{Intuitively, the noisy NMI metric $\tilde{\rho}$ serves as a predictor of whether the column splitting in the \kw{ParentGen} procedure will result in a high-fidelity data partition, which is useful for planning the next operation in \kw{DecideOP}.}
Note that when either column or row splitting is infeasible, we do not need to evaluate the correlation since the operation must be the feasible one (Line 14).


\myparagraph{3. Operation decision (procedure \kw{DecideOP})}
We determine the operation \kw{op} for the parent node.
When the size of the given table is too small to allow for row splitting and the table contains more than one attribute, the operation must be column splitting (Lines 22--23).
\changemarker{However, when both column splitting and row splitting are feasible, we choose one of them according to whether the noisy NMI $\tilde{\rho}$ exceeds a predefined threshold $\alpha$ (Lines 17--21).
Intuitively, when the correlation (measured by $\tilde{\rho}$) between the vertically split subtables is too high, the information loss regarding data distribution caused by the column splitting is excessive; 
consequently, row splitting should be prioritized in this case.}

\myparagraph{4. Budget allocation (procedure \kw{AllcBudgetOP})}
We allocate privacy budgets $\epsilon_\kw{op}, \bar{\epsilon}$ to the parent node and its children, respectively.
In the case where the operation is to vertically split a table with only two attributes, no privacy budget should be consumed for the parent node because there is only one possible column partition (Lines 26--27).
In other cases, we uniformly allocate privacy budgets among all nodes of the SPN tree.
That is, we divide the privacy budget $\epsilon$ by a scale metric $\sigma(T)$ to determine the privacy budget $\epsilon_\kw{eval} + \epsilon_\kw{op}$ for correlation trial and parent node generation, where $\sigma(T)$ represents the maximum possible number of nodes in the SPN tree of table $T$: $$\sigma(T)=2 |T|\cdot|attr(T)| / \beta - 1.$$ 
Then, we allocate a privacy budget $\epsilon_\kw{eval} = \epsilon / \sigma(T)\cdot \gamma_1$ (Line 10) for correlation trial and $\epsilon_\kw{op}= \epsilon / \sigma(T) - \epsilon_\kw{eval}$ (Line 29) for parent node generation such that $\epsilon_\kw{eval} + \epsilon_\kw{op} = \epsilon / \sigma(T)$, where $\gamma_1 \in [0, 1]$. 


\begin{algorithm}[t]
\small
    \caption{Private Fanout Construction $\kw{PrivFanout}$}
    \label{alg:fan}
    \Input{\small{Table $T_i$, SPN $t_i$, foreign key $FK_{i, j}, $privacy budget $\epsilon$}}
    \Output{Modified SPN $t_i'$}
    $\mathcal{A} \leftarrow$ Find the attribute with the most leaf nodes in $t_i$\;
    $t_i'\leftarrow$ Make a copy of $t_i$\;
    \For{each leaf node $\kw{leaf}_\mathcal{A}$ of $t_i$ of $\mathcal{A}$}{
        \For{each value $fk$ of $FK_{i,j}$}{
            $\mathcal{F}_{i,j}[fk] \leftarrow$ count the rows in $\kw{leaf}_\mathcal{A}$ for which the foreign key $FK_{i,j}$ equals $fk$\;
        }
        $\widetilde{\mathcal{F}}_{i,j} \leftarrow \mathcal{F}_{i,j} + \kw{Lap}(\frac{\Delta(\mathcal{F}_{i,j})}{\epsilon})$\;
        $\kw{leaf}_{FK_{i, j}} \leftarrow \kw{LeafNode}(\widetilde{\mathcal{F}}_{i,j})$\;
        $\kw{parent} \leftarrow \kw{ProdNode}(\mathcal{A}, FK_{i, j})$\;
        \small{Replace $\kw{leaf}_\mathcal{A}$ in $t_i'$ with subtree $(\kw{parent}, \kw{leaf}_\mathcal{A}, \kw{leaf}_{FK_{i, j}})$}
    }
\end{algorithm}

\subsection{Private Fanout Construction}
After constructing SPNs for all private tables in the input database, we follow prior work~\cite{yang2022sam, yang2020neurocard} to utilize \textit{fanout distribution}, i.e., the distribution of the foreign key in the referencing table, to model their primary-foreign key references.
However, to maintain fanout distributions for all joined keys, prior work~\cite{yang2022sam, yang2020neurocard} uses a full outer join of all the tables, which causes large space overhead.

To address this issue, we propose Algorithm \ref{alg:fan}, which models primary-foreign key references by augmenting SPNs with leaf nodes that store fanout distributions.
Specifically, to construct the reference relation between tables $T_i$ and $T_j$ where $T_i$ refers to $T_j$, we first find a certain attribute $\mathcal{A}$ and traverse all its leaf nodes.
For each leaf node $\kw{leaf}_\mathcal{A}$, we identify the rows in $T_i$ whose column of attribute $\mathcal{A}$ is stored in $\kw{leaf}_\mathcal{A}$ and count the fanout frequencies of the foreign key $FK_{i,j}$ w.r.t. the identified rows in a \textit{fanout table} $\mathcal{F}_{i,j}$ (Lines 3--5).
Then, the fanout table is perturbed by adding some Laplace noise and stored in a new leaf node $\kw{leaf}_{FK_{i, j}}$ (Lines 6--7).
Finally, we replace $\kw{leaf}_\mathcal{A}$ with a product node that connects $\kw{leaf}_\mathcal{A}$ and $\kw{leaf}_{FK_{i, j}}$ (Lines 8--9).
Additionally, to minimize the impact of Laplace noise on the fanout distribution, we require attribute $\mathcal{A}$ to have the largest leaf nodes. 
Intuitively, the overall fanout distribution of a foreign key is the average of the distributions stored in its leaf nodes. 
Thus, the variance caused by the Laplace noise decreases as the number of leaf nodes increases.

\begin{figure}[t]
  \centering

  \includegraphics[width=1.0\linewidth]{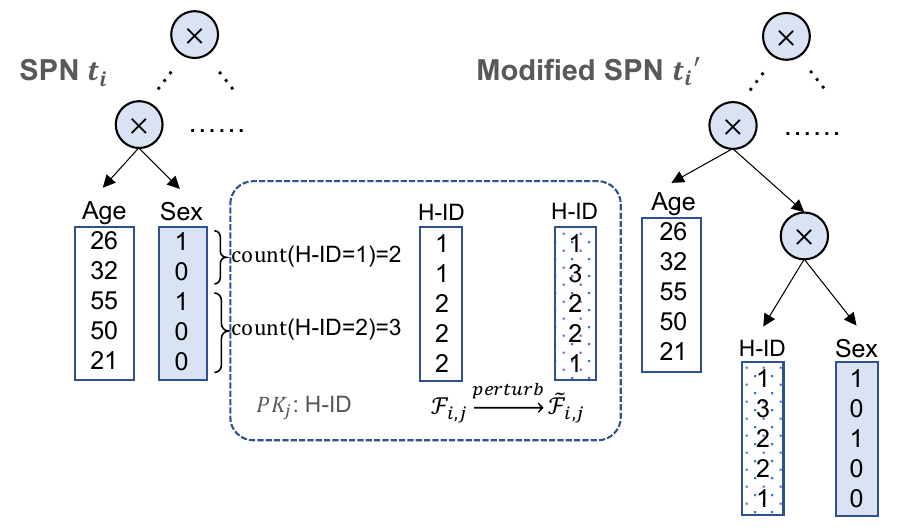}
   \caption{Fanout table construction.}
  \label{fig:modified-SPN}
\end{figure}

\begin{example}
    Figure \ref{fig:modified-SPN} shows an example of how PrivBench augments the original SPN $t_i$ with the leaf nodes of the foreign key \kw{H-ID}.
    First, we identify that the attribute \kw{Sex} possesses the most leaf nodes in $t_i$.
    For each leaf node of attribute \kw{Sex}, we locate the corresponding foreign keys using the row indices.
    Next, we compute the fanout table, perturb the fanouts, and store the results in a new leaf node.
    Finally, the new leaf node is added to the SPN as a sibling to the leaf node of attribute \kw{Sex} and as the child of a newly added product node, resulting in a modified SPN $t_i'$.
\end{example}

\subsection{SPN-Based Database Synthesis}
Algorithm \ref{alg:SampleDataFromSPN} shows how we synthesize a table $\widehat{T}_i$ given an SPN $t_i$.
\changemarker{If the root node of $t_i$ is a leaf node, we synthesize a table $\widehat{T}_i$ by histogram sampling.
Specifically, each column $\widehat{T}_i$ is sampled from the marginal distribution represented by the perturbed histogram or fanout table stored in the leaf (Lines 1--2), and the size of $\widehat{T}_i$ matches the total count of the histogram or fanout table.
}
If it is a sum (product, resp.) node, we recursively apply the \kw{SampleDataFromSPN} procedure to the left and right children of $t_i$ and vertically (horizontally, resp.) concatenate the returned tables $\widehat{T}_L$ and $\widehat{T}_R$ (Lines 4--9).
Finally, we obtain a table $\widehat{T}_i$ that has a data distribution similar to that of the private table $T_i$.
By synthesizing tables for all modified SPNs, we obtain a synthetic database $\widehat{D}$.
\changemarker{Since the input SPN $t_i$ is constructed in a differentially private manner, sampling data from it does not consume any privacy budget according to the post-processing property of DP~\cite{dwork2006differential}.}

\begin{algorithm}[t]
\small
    \caption{SampleDataFromSPN}
    \label{alg:SampleDataFromSPN}
    \Input{SPN $t_i$}
    \Output{Synthetic table $\widehat{T}_i$}    
    \If{$t_i.\kw{root}$ is a leaf node}{
        \small{$\widehat{T}_i \leftarrow$ \changemarker{Sample data from marginal distribution represented by histogram $t_i.\kw{\widetilde{his}}$ or fanout table $t_i.\widetilde{\mathcal{F}}_{i,j}$}}
    }
    
    \Else{
        $\widehat{T}_L \leftarrow \kw{SampleDataFromSPN}(t_i.\kw{left})$\;
        $\widehat{T}_R \leftarrow \kw{SampleDataFromSPN}(t_i.\kw{right})$\;
        \If{$t_i.\kw{root}$ is a sum node}{
            $\widehat{T}_i \leftarrow$ Vertically concatenate $\widehat{T}_L$ and $\widehat{T}_R$\;
        }
        \If{$t_i.\kw{root}$ is a product node}{
            $\widehat{T}_i \leftarrow$ Horizontally concatenate $\widehat{T}_L$ and $\widehat{T}_R$\;
        }
    }
    \Return $\widehat{T}_i$
\end{algorithm}



\section{Theoretical Analysis}
\label{sec:privacy}
In this section, we provide a rigorous analysis of the privacy guarantee and time complexity of PrivBench.
All the missing proofs can be found in Appendix~\ref{sec:appendix}.

\subsection{Privacy Analysis on \kw{PrivSPN}}
\label{sec:dp:privspn}
In this subsection, we first prove that the \kw{Planning}, \kw{ParentGen}, and \kw{ChildrenGen} procedures satisfy DP, and then prove that \kw{PrivSPN} therefore achieves DP.

\subsubsection{Analysis of \kw{Planning}}
The \kw{Planning} procedure consists of the \kw{CorrTrial}, \kw{DecideOP}, and \kw{AllocBudgetOP} subprocedures.
If the subprocedures satisfy DP, we can conclude that \kw{Planning} ensures DP according to the sequential composition theorem of DP~\cite{dwork2006differential}.

\myparagraph{\kw{CorrTrial}}
(1) When $|T| \geq 2\beta$ and $|attr(T)| > 1$, the \kw{CorrTrial} procedure employs the $(\epsilon_\kw{eval}\cdot\gamma_2)$-DP mechanism $\kw{ColSplit}(T, \epsilon_\kw{eval}\cdot\gamma_2)$ to compute a perturbed partition $\big(\widetilde{\mathcal{S}}_L, \widetilde{\mathcal{S}}_R\big)$.
Then, it calculates the NMI of $(\widetilde{\mathcal{S}}_L, \widetilde{\mathcal{S}}_R)$ and injects the Laplace noise $\kw{Lap}(\frac{\Delta(\kw{NMI})}{\epsilon_\kw{eval} \cdot (1-\gamma_2)})$ into it, which ensures $\big(\epsilon_\kw{eval} \cdot (1-\gamma_2)\big)$-DP.
Consequently, according to the sequential composition theorem of DP~\cite{dwork2006differential}, $\kw{CorrTrial}(T, \epsilon)$ satisfies $\epsilon_\kw{eval}$-DP.
(2) When $|T| < 2\beta$ or $|attr(T)| = 1$, since varying the value of any row in $T$ does not have any impact on the output $(\tilde{\rho}, \epsilon_\kw{eval})$, $\kw{CorrTrial}(T, \epsilon)$ satisfies $0$-DP according to Definition \ref{def:table-dp}.
Note that we adopt the bounded DP interpretation;
in the case of unbounded DP, adding/removing a row from $T$ causes the procedure to take the other conditional branch (i.e., Lines 9--12), thereby breaking the DP guarantee.

\myparagraph{\kw{DecideOP} and \kw{AllocBudgetOP}}
Similar to the second case of \kw{CorrTrial}, for both procedures \kw{DecideOP} and \kw{AllocBudgetOP}, because changing the value of any row in $T$ does not affect the table size $|T|$ and the number of attributes $|attr(T)|$, it also does not impact their outputs.
Therefore, both procedures ensure $0$-DP.

\myparagraph{\kw{Planning}}
When $|attr(T)|=1$, \kw{Planning} satisfies $0$-DP since any change to $T$ does not impact the output (Lines 1--2);
otherwise, as it sequentially combines \kw{CorrTrial}, \kw{DecideOP}, and \kw{AllocBudgetOP}, $\kw{Planning}(T, \epsilon)$ satisfies $\epsilon_\kw{eval}$-DP according to the sequential composition theorem of DP~\cite{dwork2006differential}, where the value of $\epsilon_\kw{eval}$ differs in different cases (see Lines 10 and 14).
 
\begin{lemma}
\label{lem:Planning}
    If $|attr(T)| > 1$, $\kw{Planning}(T, \epsilon)$ satisfies table-level $\big(\epsilon\cdot \gamma_1 / \sigma(T)\big)$-DP;
    otherwise, it satisfies table-level $0$-DP.
\end{lemma}

\subsubsection{Analysis of \kw{ParentGen}}
Then, we show that in any case of the given operation $\kw{op}$, \kw{ParentGen} achieves DP.
(1) When $\kw{op} = \kw{OP.LEAF}$, we ensure DP using the Laplace mechanism.
That is, we inject Laplacian noise $\kw{Lap}(\frac{\Delta(\kw{his})}{\epsilon_\kw{op}})$ into the histogram $\kw{his}(T)$, which satisfies $\epsilon_\kw{op}$-DP.
(2) When $\kw{op} = \kw{OP.SUM}$, we call $\kw{RowSplit}(T, \epsilon_\kw{op})$ to generate a perturbed row partition $(\widetilde{S}_L, \widetilde{S}_R)$, which achieves $\epsilon_\kw{op}$-DP.
(3) When $\kw{op} = \kw{OP.PRODUCT}$, our column splitting mechanism $\kw{ColSplit}(T, \epsilon_\kw{op})$ is essentially an instance of the exponential mechanism~\cite{dwork2014algorithmic}, thereby guaranteeing $\epsilon_\kw{op}$-DP.

\begin{lemma}
\label{lem:ParentGen}
    $\kw{ParentGen}(T,\kw{op},\epsilon_\kw{op})$ satisfies tabel-level $\epsilon_\kw{op}$-DP.
\end{lemma}

\subsubsection{Analysis of \kw{ChildrenGen}}
Next, we analyze the DP guarantee of $\kw{ChildrenGen}(T, \kw{op}, \widetilde{\mathcal{S}}_L, \widetilde{\mathcal{S}}_R, \bar{\epsilon})$ in different cases.
(1) When $\kw{op} = \kw{OP.LEAF}$, since \kw{ChildrenGen} always returns two null children, it must satisfy $0$-DP.
Therefore, in this case, $\kw{PrivSPN(T, \epsilon)}$ achieves $\epsilon$-DP.
(2) When $\kw{op} = \kw{OP.SUM}$, Theorem \ref{thm:sumnode} demonstrates a parallel composition property of DP in the case of row splitting: If constructing the left and right subtrees satisfies $\epsilon_L$-DP and $\epsilon_R$-DP, respectively, then the entire process of constructing both subtrees satisfies $\max\{\epsilon_L, \epsilon_R\}$-DP.
Therefore, to optimize the utility of the subtrees, we maximize each subtree's privacy budget by setting $\epsilon_L = \epsilon_R = \bar{\epsilon}$.
Note that Theorem \ref{thm:sumnode} differs from the celebrated parallel composition theorem \cite{dwork2006differential}:
their theorem assumes unbounded DP while our theorem considers bounded DP.
(3) When $\kw{op} = \kw{OP.PRODUCT}$, we allocate privacy budgets based on the scales of the subtrees (Line 18).
Intuitively, a subtree with a larger scale should be assigned with a larger privacy budget to balance their utility.
Thus, their privacy budgets $\epsilon_L, \epsilon_R$ are proportional to their scales $\sigma(T[\widetilde{S}_L]), \sigma(T[\widetilde{S}_R])$.
Note that publishing the scales does not consume any privacy budget since varying any row of $T$ does not change $\sigma(T[\widetilde{S}_L])$ and $\sigma(T[\widetilde{S}_R])$.

\begin{theorem}[Parallel composition under bounded DP]
\label{thm:sumnode}
    Given a row partition $(\mathcal{S}_1,\dots, \mathcal{S}_K)$, publishing $\mathcal{M}_1(T[\mathcal{S}_1]),\dots,\mathcal{M}_K(T[\mathcal{S}_K])$ satisfies table-level $\max\set{\epsilon_1,\dots, \epsilon_K}$-DP, where $\mathcal{S}_k$ is a subset of row indices and $\mathcal{M}_k$ is a table-level $\epsilon_k$-DP algorithm, $\forall k \in [K]$.
\end{theorem}


Consequently, when $\kw{op} = \kw{OP.SUM}$ or $\kw{op} = \kw{OP.PRODUCT}$, we can prove the DP guarantee by mathematical induction:
when $\kw{PrivSPN}(T[\widetilde{S}_L], \epsilon_L)$ (resp. $\kw{PrivSPN}(T[\widetilde{S}_R], \epsilon_R)$) returns a subtree with only a leaf node, it satisfies $\epsilon_L$-DP (resp. $\epsilon_R$-DP);
otherwise, assuming $\kw{PrivSPN}(T[\widetilde{S}_L], \epsilon_L)$ (resp. $\kw{PrivSPN}(T[\widetilde{S}_R], \epsilon_R)$) satisfies $\epsilon_L$-DP (resp. $\epsilon_R$-DP), according to the sequential composition theorem~\cite{dwork2006differential} or Theorem \ref{thm:sumnode}, we conclude that $\kw{ChildrenGen}(T, \kw{op}, \widetilde{\mathcal{S}}_L, \widetilde{\mathcal{S}}_R, \bar{\epsilon})$ satisfies $\bar{\epsilon}$-DP, where $\bar{\epsilon} = \epsilon_L + \epsilon_R$ or $\bar{\epsilon} = \max\{\epsilon_L, \epsilon_R\}$.

\begin{lemma}
\label{lem:ChildrenGen}
    When $\kw{op} = \kw{OP.LEAF}$, $\kw{ChildrenGen}(T, \kw{op}, \widetilde{\mathcal{S}}_L, \widetilde{\mathcal{S}}_R, \bar{\epsilon})$ satisfies table-level $0$-DP;
    otherwise, it satisfies table-level $\bar{\epsilon}$-DP.
\end{lemma}

\subsubsection{Analysis of \kw{PrivSPN}}
Based on Lemmas \ref{lem:Planning}, \ref{lem:ParentGen}, and \ref{lem:ChildrenGen}, we conclude the DP guarantee of \kw{PrivSPN} through sequential composition.

\begin{lemma}
    $\kw{PrivSPN}(T, \epsilon)$ satisfies table-level $\epsilon$-DP.
\end{lemma}

In addition, since \kw{PrivSPN} allocates the total privacy budget $\epsilon$ from the root to the leaf nodes in a top-down manner, it may result in insufficient budgets allocated to the leaves.
Consequently, a practical question arises: If we would like to guarantee a privacy budget of $\epsilon_\kw{leaf}$ for each leaf node, how large should the total privacy budget $\epsilon$ be?
Theorem \ref{thm:down-top} answers this question: 
A privacy budget $\epsilon= \big(\epsilon_\kw{leaf}\cdot \sigma(T)\big)$ is sufficient.
This also explains why the privacy budgets $\epsilon_L$ and $\epsilon_R$ should be proportional to the corresponding scales $\sigma(T[\widetilde{S}_L])$ and $\sigma(T[\widetilde{S}_R])$ in Line 18 of Algorithm \ref{alg:SPN}.

\begin{theorem}
  \label{thm:down-top}
  Consider a procedure $\kw{PrivSPN}(T, \epsilon)$ that publishes an SPN with $K$ leaf nodes $\kw{PrivSPN}(S_1, \epsilon_\kw{leaf}),\dots, \kw{PrivSPN}(S_K, \epsilon_\kw{leaf})$, where $S_k$ is a single-column subtable of table $T$, $\forall k \in [K]$.
  $\kw{PrivSPN}(T, \epsilon)$ satisfies table-level $\big(\epsilon_\kw{leaf}\cdot \sigma(T)\big)$-DP.
\end{theorem}

\subsection{Privacy Analysis on \kw{PrivFanout}}
For each leaf node $\kw{leaf}_{FK}$, \kw{PrivFanout} injects the Laplace noise into the fanout table $\mathcal{F}_{i,j}$ (Line 6), which implements an $\epsilon$-differentially private Laplace mechanism.
Then, according to Theorem \ref{thm:sumnode}, publishing all the leaf nodes $\kw{leaf}_{FK}$ with perturbed fanout tables also satisfies $\epsilon$-DP due to the property of parallel composition.
Therefore, \kw{PrivFanout} ensures DP.

\begin{lemma}
    $\kw{PrivFanout}(T_i, t_i, FK_{i,j}, \epsilon)$ satisfies table-level $\epsilon$-DP.
\end{lemma}

\subsection{Privacy Analysis on \kw{PrivBench}}
\label{sec:analysis_privbench}


Given that both \kw{PrivSPN} and \kw{PrivFanout} satisfy DP, the DP guarantee of the \kw{PrivBench} algorithm is established according to Theorem \ref{thm:table-to-db-dp}.
Note that to measure the indistinguishability level of database-level DP, we cannot simply accumulate the privacy budgets $\epsilon_i$ assigned to each private table $T_i \in D$ for table-level DP.
Intuitively, while table-level DP only guarantees the indistinguishability of neighboring tables, database-level DP requires indistinguishability of two tables that differ in at most $\tau_i$ tuples because those tuples may depend on the same tuple in the primary private table.
Consequently, as shown in Corollary \ref{coro:privbench}, the indistinguishability levels of procedures $\kw{PrivSPN}$ and $\kw{PrivFanout}$ at the table level should be reduced by factor $\tau_i$ for database-level DP.
Therefore, to ensure database-level $\epsilon$-DP for \kw{PrivBench}, we can allocate the total privacy budget $\epsilon$ as follows: 
$$\epsilon^s_i=\frac{\epsilon \cdot \gamma}{\tau_1 + \dots +\tau_n}, \quad \epsilon^f_i=\frac{\epsilon \cdot (1-\gamma)}{\tau_1 \cdot |FK(T_1)| + \dots + \tau_n \cdot |FK(T_n)|},$$
where $\gamma\in[0,1]$ is the ratio of privacy budget allocated for SPN construction, and $|FK(T_i)|$ is the number of foreign keys in table $T_i$.
In our experiments, we follow the above setting to allocate privacy budgets $\epsilon^s_i, \epsilon^s_f$ in PrivBench.


\begin{theorem}
\label{thm:table-to-db-dp}
    Given that $\mathcal{M}_i$ satisfies table-level $\epsilon_i$-DP for all $i\in[n]$, the composite mechanism $\mathcal{M}(D)=\big(\mathcal{M}_1(T_1),...,\mathcal{M}_n(T_n)\big)$ satisfies database-level $(\sum_{i\in[n]}\tau_i\cdot\epsilon_i)$-DP.
\end{theorem}

\begin{corollary}[of Theorem \ref{thm:table-to-db-dp}]
\label{coro:privbench}
    \kw{PrivBench} satisfies database-level $(\sum_{i\in[n]}\tau_i\cdot\epsilon^s_i+\sum_{T_i \text{ refers to } T_j} \tau_i\cdot\epsilon^f_i)$-DP.
\end{corollary}

\subsection{Time Complexity Analysis}
\label{sec:time_complexity}
We show that $\ours(D)$ completes in polynomial time.
(1) To construct a private SPN $t_i$ for a table $T_i$, \kw{PrivSPN} needs to recursively call itself to generate a full binary tree, where each node requires one call of \kw{PrivSPN}.
Because each leaf node stores a single-column subtable of $T_i$ with at least $\beta$ rows and one attribute, $t_i$ can have at most $(|attr(T_i)| \cdot |T_i|/\beta)$ leaf nodes.
Consequently, $t_i$ can have at most $(2|attr(T_i)|\cdot |T|/\beta -1)$ nodes, which can be computed by $\kw{PrivSPN}(T_i, \epsilon^s_{i})$ with $\mathcal{O}(|attr(T_i)|\cdot |T_i|)$ recursive calls.
Then, we can easily observe that the non-recursive work including running the \kw{Planning} and \kw{ParentGen} procedures can finish in polynomial time.
Therefore, $\kw{PrivSPN}(T_i, \epsilon^s_{i})$ completes in polynomial time.
(2) To construct a fanout table for each pair of referential tables $T_i, T_j$, \kw{PrivFanout} needs to enumerate at most $(|attr(T_i)| \cdot |T_i|/\beta)$ leaf nodes of $t_i$ to replace each leaf with a subtree, which finishes in $\mathcal{O}(|attr(T_i)|\cdot |T_i|)$ time.
(3) For each modified SPN $t_i'$, $\kw{SampleDataFromSPN}(t_i')$ requires $\mathcal{O}(|attr(T_i)|\cdot |T_i|)$ recursive calls and some polynomial-time non-recursive work;
thus, it finishes in $\mathcal{O}(|attr(T_i)|\cdot |T_i|)$ time. 
Given that \kw{PrivSPN}, \kw{PrivFanout}, and \kw{SampleDataFromSPN} finish in polynomial time, PrivBench completes in polynomial time.
\begin{lemma}
    Given a database $D=\{T_1,\dots, T_n\}$, $\kw{PrivBench}(D)$ finishes in $\mathcal{O}\big(\sum_{i\in[n]} \kw{poly} (|attr(T_i)|, |T_i|)\big)$ time.
\end{lemma}

\captionsetup[subfigure]{ singlelinecheck=false, font=footnotesize, justification=centering}
\begin{figure}[t]
  \centering
  \subfloat[Selectivity over Original Data]{
        \includegraphics[width=0.4\linewidth]{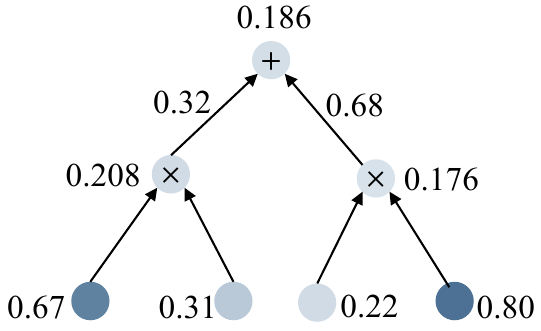}
        \label{fig:QHM:l}

  }
  \quad
  \subfloat[Selectivity over Generated Data]{
        \includegraphics[width=0.5\linewidth]{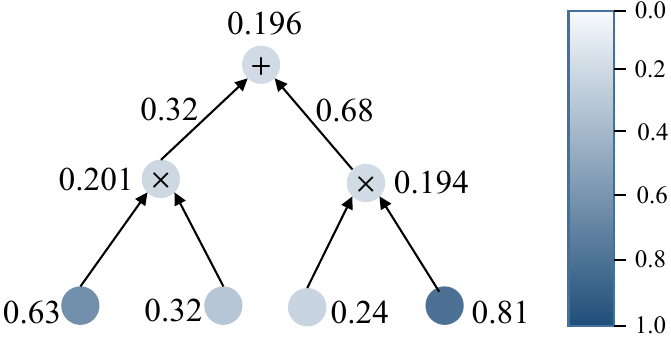}
        \label{fig:QHM:r}        
  }
  
  \caption{An example of Query Heat Map for a given query. The edge values indicate the proportions of rows assigned to each subtree of a sum node, while the node values represent the selectivities for the given query.}
  \label{fig:QHM}
\end{figure}

\section{Discussion}
\label{sec:discussion}
In this section, we discuss several practical issues regarding PrivBench.

\subsubsection*{Privacy-preserving/secure data valuation}
In recent years, research on data trading has been gaining increasing attention in the database community~\cite{li2014theory, koutris2015query, chen2019towards, zheng2020money, liu2021dealer, zheng2022fl, zheng2023secure}. 
A key research question is how to address the Arrow Information Paradox for privacy-preserving/secure data valuation. 
Specifically, since the utility of data is task-specific, data buyers often need to test the data to determine its value before reaching a deal. 
However, once a buyer gains access to the data, the seller faces the risk of the data being replicated without payment. 
To tackle this issue for trading machine learning models, Zheng et al.~\cite{zheng2023secure} utilized homomorphic encryption (HE) to enable buyers to evaluate the value of data in an encrypted environment. 
PrivBench can also address the paradox by allowing buyers to assess the value of data based on synthetic databases while ensuring the privacy and security of the original data. 
Compared to HE-based methods, PrivBench-empowered data valuation, though sacrificing some accuracy due to the injection of DP noise, offers better computational efficiency and facilitates trustworthy, real-time data trading.

\subsubsection*{Parameter selection}
\changemarker{
Here, we discuss how to configure the parameters of PrivBench.
The privacy budget $\epsilon$ is a general parameter in DP, and various methods have been proposed for selecting it, such as economic methods \cite{hsu2014differential} or voting mechanisms \cite{kohli2018epsilon}.
The other parameters, including $\alpha$, $\beta$, $\gamma$, $\gamma_1$, and $\gamma_2$, however, are unique to PrivBench and can be determined empirically.
In practice, since DP assumes that the benchmark publisher is trusted, they can access both the original and synthetic data to evaluate the latter's quality in terms of data distribution similarity and query runtime similarity. 
Therefore, every time the publisher releases a benchmark, they can incrementally optimize the parameters based on experimental results to improve the performance of future benchmarks. 
Additionally, in our experiments, we apply the same parameters to all datasets. 
The results demonstrate that PrivBench consistently outperforms other methods, which suggests the transferability of the parameter settings among different datasets. 
In other words, even when adopting the same parameters for different benchmarking tasks, publishers still can achieve high-quality synthetic databases.
}


\subsubsection*{Visualization}
Visualization enables us to intuitively evaluate whether the data generated by PrivBench closely approximates the original data as a benchmark. 
As shown in Figure \ref{fig:QHM}, after synthesizing data using PrivBench, we can obtain a pair of SPNs built on the original and synthetic data, respectively. 
After testing a bundle of queries, we can determine the selectivity value for each leaf node. 
Subsequently, we can color the leaf nodes based on the magnitude of the selectivity value, with darker colors indicating higher values, creating a colored SPN. Ultimately, by comparing the colored SPN of the original data with that of the synthetic data, we can assess whether the two sets of data perform similarly in benchmarking.

\begin{table*}[ht]
\footnotesize
\caption{Summary of datasets, workloads, and baselines.}
\label{tab:datasets}
\begin{tabular}{ccccc}
\hline
\textbf{Dataset} & \textbf{\#(Private) Tables} & \textbf{(\#Rows, \#Cols.) of Private Table} & \textbf{Workload} & \textbf{Baseline} \\ \hline
Adult & 1 (1) & (45222, 15) & SAM-1000 & \begin{tabular}[c]{@{}c@{}}PrivSyn, PrivBayes, AIM, DataSynthesizer, MST, Exponential-PreFair,\\ Greedy-PreFair, DPGAN (GPU-accelerated), PrivMRF (GPU-accelerated)\end{tabular} \\ \hline
California & 2 (2) & \begin{tabular}[c]{@{}c@{}}Primary: (616115, 10)\\ Secondary: (1690642, 16)\end{tabular} & California-400 & PrivLava (GPU-accelerated) \\ \hline
JOB-light & 6 (2) & \begin{tabular}[c]{@{}c@{}}Primary: (2283757, 2)\\ Secondary: (35824707, 3)\end{tabular} & JOB-229, MSCN-226 & PrivLava (GPU-accelerated) \\ \hline
\end{tabular}
\end{table*}

\section{Experimental Evaluation}
\label{sec:exp}

\subsection{Experimental Settings}\label{sec:expsetting}
\myparagraph{Research questions}
In this section, we conduct extensive experiments to answer the following research questions.
\begin{itemize}[leftmargin=*]
    \item \textbf{Data distribution similarity}: Is the database synthesized by PrivBench closer to the original database in data distribution?
    \item \textbf{Query runtime similarity}: For executing query workloads, is the database synthesized by PrivBench closer to the original database in terms of query runtime performance?
    \item \changemarker{\textbf{Synthesis time}: Can PrivBench synthesize databases efficiently?}
\end{itemize}

\myparagraph{Datasets, query workloads, and baselines} 
We use the following datasets, query workloads, and baselines to verify the performance of PrivBench, which are summarized in Table~\ref{tab:datasets}.

The \textit{Adult} dataset~\cite{kohavi1996scaling} contains a single private table of census information about individuals. 
We execute \textit{SAM-1000}~\cite{yang2022sam}, a workload with 1000 randomly generated queries, to evaluate the query runtime performance of PrivBench and baselines.
The baselines cover mainstream differentially private single-relation data synthesis methods, including DataSynthesizer~\cite{ping2017datasynthesizer}, PrivBayes~\cite{zhang2017privbayes}, PrivSyn~\cite{zhang2021privsyn},   MST~\cite{mckenna2021winning}, PrivMRF~\cite{cai2021privmrf}, AIM~\cite{mckenna2022aim}, Exponential-PreFair~\cite{pujol2022prefair}, Greedy-PreFair~\cite{pujol2022prefair}, and DPGAN~\cite{liu2023tabular}.
Note that PrivMRF is the single-relation version of PrivLava~\cite{cai2023privlava}.

The \textit{California} dataset~\cite{cai2023privlava} consists of two private tables about household information.
We randomly generate 400 queries following the method from prior work~\cite{kipf2018learned} to create the query workload for this dataset, named \textit{California-400}.
We use PrivLava~\cite{cai2023privlava}  as the sole baseline, as it is the only SOTA method that supports multi-relation database synthesis.

The \textit{JOB-light} dataset~\cite{yang2022sam} includes six tables related to movies, which are extracted from the \textit{Internet Movie Database (IMDB \cite{imdb})}.
We designate two of these tables as private, while the remaining four are public.
Consequently, we generate synthetic tables only for the two private tables and combine these synthetic tables with the public tables into a single database for testing.
We adopt the \textit{MSCN} and \textit{JOB} query workloads from prior work~\cite{yang2022sam} and extract subsets of 226 and 229 queries, respectively, that involve operations over the two private tables.
These subsets, named \textit{MSCN-226} and \textit{JOB-229}, serve as the query workloads for the Job-light dataset.
PrivLava~\cite{cai2023privlava} serves as the baseline for multi-relation synthesis.

\begin{table*}[t]
\small
    \centering
    \caption{Performance of data synthesis methods on the Adult dataset with the SAM-1000 query workload.}
\begin{tabular}{l|ccc|cccc|cccc|cc}
\hline
\multicolumn{1}{c|}{\multirow{2}{*}{\textbf{Model}}} & \multicolumn{3}{c|}{\textbf{KLD}} & \multicolumn{4}{c|}{\textbf{Q-error}} & \multicolumn{4}{c|}{\textbf{Query runtime discrepancy (\%)}} & \multicolumn{2}{c}{\textbf{Synthesis time (s)}} \\ \cline{2-14} 
\multicolumn{1}{c|}{} & \textbf{2-way} & \textbf{3-way} & \textbf{4-way} & \textbf{Mean} & \textbf{Median} & \textbf{75th} & \textbf{MAX} & \textbf{Mean} & \textbf{Median} & \textbf{75th} & \textbf{MAX} & \textbf{Learning} & \textbf{Inference} \\ \hline
PrivSyn & 13.01 & 14.68 & 15.03 & 2.56 & 1.30 & 1.60 & 22.79 & 8.76 & 5.86 & 13.83 & 47.12 & 72.61 & 10.59 \\
PrivBayes & 10.37 & 12.64 & 13.74 & 1.41 & 1.31 & 1.55 & 3.45 & 9.67 & 3.71 & 13.95 & 95.82 & 12.58 & 5.32 \\
AIM & 7.75 & 10.13 & 11.63 & \textbf{1.34} & \textbf{1.24} & \textbf{1.50} & \textbf{2.98} & 40.60 & 39.48 & 41.97 & 77.04 & 1865 & 144.17 \\
DataSynthesizer & 6.92 & 9.84 & 11.94 & 1.40 & 1.26 & 1.58 & 3.55 & 8.36 & 3.08 & 11.09 & 62.42 & \textbf{0.12} & 0.51 \\
MST & 5.24 & 7.07 & 8.53 & 1.38 & \textbf{1.24} & 1.52 & 5.71 & 14.73 & 14.47 & 16.57 & 35.34 & 178.80 & 0.11 \\
Exponential-PreFair & 5.12 & 6.98 & 8.48 & 1.43 & 1.30 & 1.61 & 7.2 & 18.50 & 17.93 & 20.23 & 37.78 & 1736 & 0.14 \\
Greedy-PreFair & 4.73 & 6.52 & 8.02 & 1.44 & 1.29 & 1.60 & 17.16 & 4.69 & 3.73 & 6.21 & 26.54 & 180.33 & \textbf{0.10} \\
DPGAN & 4.56 & 6.68 & 8.55 & 1.37 & 1.27 & 1.52 & 3.53 & 8.18 & 6.28 & 10.77 & 57.68 & 746.53 & 0.49 \\
PrivMRF (PrivLava) & 4.55 & 6.29 & 7.79 & 1.38 & 1.25 & 1.52 & 4.19 & 3.11 & 2.61 & 4.29 & 21.6 & 1703 & 41.36 \\
PrivBench & \textbf{1.71} & \textbf{2.79} & \textbf{4.05} & 1.40 & 1.25 & 1.56 & 7.78 & \textbf{1.84} & \textbf{1.48} & \textbf{2.38} & \textbf{12.96} & 2.35 & 0.13 \\ \hline
\end{tabular}
\label{tab:adult}
\end{table*}

\begin{table*}[ht]
\small
\centering
\caption{Performance of database synthesis methods on the California dataset with the California-400 query workload.}
\begin{tabular}{|c|cccccc|cccccc|cccccc|}
\hline
\multirow{3}{*}{\textbf{Model}} & \multicolumn{6}{c|}{\textbf{2-way KLD}} & \multicolumn{6}{c|}{\textbf{3-way KLD}} & \multicolumn{6}{c|}{\textbf{3-way KLD}} \\ \cline{2-19} 
 & \multicolumn{6}{c|}{\textbf{Privacy budget $\epsilon$}} & \multicolumn{6}{c|}{\textbf{Privacy budget $\epsilon$}} & \multicolumn{6}{c|}{\textbf{Privacy budget $\epsilon$}} \\ \cline{2-19} 
 & \textbf{0.1} & \textbf{0.2} & \textbf{0.4} & \textbf{0.8} & \textbf{1.6} & \textbf{3.2} & \textbf{0.1} & \textbf{0.2} & \textbf{0.4} & \textbf{0.8} & \textbf{1.6} & \textbf{3.2} & \textbf{0.1} & \textbf{0.2} & \textbf{0.4} & \textbf{0.8} & \textbf{1.6} & \textbf{3.2} \\ \hline
\multicolumn{1}{|l|}{PrivLava} & 2.30 & 1.28 & 1.03 & 0.93 & 0.88 & 0.88 & 3.93 & 2.21 & 1.77 & 1.56 & 1.48 & 1.46 & 5.50 & 3.20 & 2.59 & 2.27 & 2.14 & 2.09 \\
\multicolumn{1}{|l|}{PrivBench} & 0.87 & 0.77 & 0.69 & 0.63 & 0.60 & 0.58 & 1.85 & 1.74 & 1.62 & 1.54 & 1.51 & 1.50 & 2.94 & 2.84 & 2.67 & 2.58 & 2.56 & 2.55 \\ \hline
\multirow{3}{*}{\textbf{Model}} & \multicolumn{6}{c|}{\textbf{Q-error}} & \multicolumn{6}{c|}{\textbf{Query runtime discrepancy (\%)}} & \multicolumn{6}{c|}{\textbf{Synthesis time (s)}} \\ \cline{2-19} 
 & \multicolumn{6}{c|}{\textbf{Privacy budget $\epsilon$}} & \multicolumn{6}{c|}{\textbf{Privacy budget $\epsilon$}} & \multicolumn{3}{c}{\multirow{2}{*}{\textbf{Learning}}} & \multicolumn{3}{c|}{\multirow{2}{*}{\textbf{Inference}}} \\ \cline{2-13}
 & \textbf{0.1} & \textbf{0.2} & \textbf{0.4} & \textbf{0.8} & \textbf{1.6} & \textbf{3.2} & \textbf{0.1} & \textbf{0.2} & \textbf{0.4} & \textbf{0.8} & \textbf{1.6} & \textbf{3.2} & \multicolumn{3}{c}{} & \multicolumn{3}{c|}{} \\ \hline
\multicolumn{1}{|l|}{PrivLava} & 791.15 & 66.46 & 8.31 & 2.26 & 1.45 & 1.33 & 5.20 & 4.10 & 3.83 & 3.05 & 2.88 & 3.18 & \multicolumn{3}{c}{1520} & \multicolumn{3}{c|}{26013} \\
\multicolumn{1}{|l|}{PrivBench} & 26.07 & 5.57 & 7.17 & 5.11 & 5.18 & 5.18 & 5.70 & 5.34 & 4.59 & 4.18 & 4.26 & 4.29 & \multicolumn{3}{c}{38.31} & \multicolumn{3}{c|}{4.92} \\ \hline
\end{tabular}
\label{tab:california}
\end{table*}

\begin{table*}[ht]
\small
\centering
\caption{Performance of database synthesis methods on the JOB-light dataset with the JOB-229 and MSCN-226 query workloads.}
\begin{tabular}{|c|ccccccccccclc|}
\hline
\multirow{2}{*}{\textbf{Model}} & \multicolumn{3}{c|}{\multirow{2}{*}{\textbf{2-way KLD}}}          & \multicolumn{4}{c|}{\textbf{Q-error}}                                               & \multicolumn{3}{c|}{\textbf{Query runtime discrepancy (\%)}}                      & \multicolumn{3}{c|}{\textbf{Synthesis time}}                     \\ \cline{5-14} 
                                & \multicolumn{3}{c|}{}                                             & \multicolumn{2}{c}{\textbf{JOB-229}} & \multicolumn{2}{c|}{\textbf{MSCN-226}}       & \textbf{JOB-229} & \multicolumn{2}{c|}{\textbf{MSCN-226}}             & \multicolumn{2}{c}{\textbf{Learning}} & \textbf{Inference}       \\ \hline
\multicolumn{1}{|l|}{PrivLava}  & \multicolumn{3}{c|}{5.40}                                         & \multicolumn{2}{c}{152.05}           & \multicolumn{2}{c|}{34.52}                   & 32.77            & \multicolumn{2}{c|}{35.47}                         & \multicolumn{2}{c}{8101}              & 18104                    \\
\multicolumn{1}{|l|}{PrivBench} & \multicolumn{3}{c|}{3.16}                                         & \multicolumn{2}{c}{3.33}             & \multicolumn{2}{c|}{11.32}                   & 20.10            & \multicolumn{2}{c|}{13.24}                         & \multicolumn{2}{c}{40.72}             & 30.22                    \\ \hline
\multirow{4}{*}{\textbf{Model}} & \multicolumn{13}{c|}{\textbf{Q-error}}                                                                                                                                                                                                                                                             \\ \cline{2-14} 
                                & \multicolumn{7}{c|}{\textbf{JOB-229}}                                                                                                                   & \multicolumn{6}{c|}{\textbf{MSCN-226}}                                                                                                   \\ \cline{2-14} 
                                & \multicolumn{3}{c|}{\textbf{Cardinality}}                         & \multicolumn{4}{c|}{\textbf{No. joins}}                                             & \multicolumn{3}{c|}{\textbf{Cardinality}}                             & \multicolumn{3}{c|}{\textbf{No. joins}}                          \\ \cline{2-14} 
                                & \textbf{Low} & \textbf{Med.} & \multicolumn{1}{c|}{\textbf{High}} & \textbf{1}        & \textbf{2}       & \textbf{3} & \multicolumn{1}{c|}{\textbf{4}} & \textbf{Low}     & \textbf{Med.} & \multicolumn{1}{c|}{\textbf{High}} & \multicolumn{2}{c}{\textbf{1}}        & \textbf{2}               \\ \hline
\multicolumn{1}{|l|}{PrivLava}  & 421.59       & 13.80         & \multicolumn{1}{c|}{22.69}         & 8.73              & 15.79            & 74.49      & \multicolumn{1}{c|}{1827.51}    & 84.43            & 9.04          & \multicolumn{1}{c|}{10.43}         & \multicolumn{2}{c}{23.21}             & 39.59                    \\
\multicolumn{1}{|l|}{PrivBench} & 4.81         & 1.67          & \multicolumn{1}{c|}{3.55}          & 3.31              & 2.96             & 3.68       & \multicolumn{1}{c|}{4.51}       & 29.39            & 2.09          & \multicolumn{1}{c|}{2.62}          & \multicolumn{2}{c}{22.53}             & 6.30                     \\ \hline
\end{tabular}
\label{tab:job-light}
\end{table*}

\myparagraph{Metrics}
For the first research question, we evaluate $\lambda$-way KLD to measure data distribution similarity between original and synthetic data.
Specifically, following prior work~\cite{cai2021privmrf}, we enumerate all possible $\lambda$-way marginals for each private relation, where $\lambda\in \{2, 3, 4\}$.
For each $\lambda$-way marginal, we compute the KLD between the original and synthetic tables.
The $\lambda$-way KLD is then obtained as the average KLD across all $\lambda$-way marginals for all private relations.
To prevent infinite KLD values, a small constant of $10^{-10}$ is added to each probability distribution.
For the second research question regarding query runtime similarity, in addition to the Q-error metric, we also use the \textit{query runtime discrepancy}, which measures the average percentage difference in query runtime.
For the third research question, we report the time consumed for database synthesis.

\myparagraph{Parameters}
We set the default parameters of PrivBench as follows:
$\alpha=0.5, \beta=10000, \gamma=0.9, \gamma_1=0.5, \gamma_2=0.5$.
Then, the total privacy budget $\epsilon$ is set to $3.2$ by default, following the setting in the SOTA work, PrivLava~\cite{cai2023privlava}. 
The setting of the number of iterations $J$ for \kw{RowSplit} follows prior analysis~\cite{su2016differentially}.
Additionally, the baselines PrivSyn, MST, PrivMRF, AIM, Exponential-PreFair, Greedy-PreFair, DPGAN, and PrivLava only satisfy $(\epsilon, \delta)$-DP, where the $\delta$ parameter allows the privacy guarantee to be violated with a small probability.
We set $\delta=10^{-12}$ for all these baselines.

\myparagraph{Environment}
All experiments are implemented in Python and performed on a Linux server with an Intel(R) Core(R) Silver i9-13900K 3.0GHz CPU, an NVIDIA GeForce RTX 4090 GPU, and 64GB RAM.
The DBMS we use to test query execution is PostgreSQL 15.2 with default settings for all parameters.
We accelerate matrix computation using the GPU for the computationally inefficient baselines DPGAN, PrivMRF, and PrivLava.

\subsection{Evaluation on Distribution Similarity}\label{sec:expsingle}
\myparagraph{Single-relation synthesis}
Table~\ref{tab:adult} shows the results on the Adult dataset.
We can see that for different values of $\lambda$, PrivBench performs significantly better than all baselines on the $\lambda$-way KLD metric.
This indicates that PrivBench can synthesize single-table data with a higher fidelity of data distribution.

\myparagraph{Multi-relation synthesis}
Tables \ref{tab:california} and \ref{tab:job-light} compare multi-relation database synthesis methods on the California and JOB-light datasets.
We can observe that in various cases, the performance of PrivBench in terms of KLD is comparable to or even better than that of the SOTA method PrivLava.
Moreover, as shown in Table \ref{tab:california}, as the privacy budget decreases, the deterioration trend of the KLD metric for PrivBench is more moderate compared to PrivLava. 
This indicates that PrivBench has a greater advantage in scenarios where the original data is highly sensitive or data sharing is strictly restricted, resulting in a tiny privacy budget.

\subsection{Evaluation on Query Runtime Similarity}
\label{subsec: sim_for_tpch}

\myparagraph{Single-relation synthesis}
For the Adult dataset, as shown in Table \ref{tab:adult}, on Q-error-related metrics, both PrivBench and the SOTA methods perform exceptionally well. 
However, in terms of query runtime discrepancy, PrivBench significantly outperforms them. 
This indicates that the single-table benchmarks generated by PrivBench can better preserve similarity in query runtime performance.


\myparagraph{Multi-relation synthesis}
Table \ref{tab:california} presents the Q-error and runtime discrepancy for PrivBench and PrivLava under different privacy budgets for the California dataset. 
Similar to the case with the KLD metric, PrivBench underperforms compared to PrivLava on the Q-error metric when the privacy budget is relatively large. 
However, as the privacy budget decreases, PrivLava's Q-error performance deteriorates rapidly, while PrivBench remains relatively stable and significantly surpasses PrivLava when the privacy budget is very limited. 
This demonstrates PrivBench's robustness in data-sensitive scenarios.
Additionally, PrivBench and PrivLava both perform excellently in query runtime discrepancy, with no significant difference between them.


\changemarker{For the JOB-light dataset, we present the performance of the two synthesis methods under different types of queries in Table \ref{tab:job-light}. 
Specifically, the JOB-229 and MSCN-226 query workloads are divided into three groups based on query cardinality levels: low, medium, and high. 
The results indicate that PrivBench achieves consistently lower Q-error across all cardinality levels.
When the cardinality level is low, PrivLava's Q-error deteriorates significantly, whereas PrivBench performs markedly better than PrivLava.
Moreover, the queries are also categorized based on the number of joins in Table \ref{tab:job-light}. 
As the number of joins increases, the Q-error of PrivLava can become extremely large, while that of PrivBench remains low. 
This suggests that PrivBench may perform better on enterprise-level databases with complex internal dependencies among tables.}

\subsection{Evaluation on Synthesis Time}
\myparagraph{Single-relation synthesis}
\changemarker{
We report the time required to synthesize the Adult dataset in Table \ref{tab:adult}. 
The synthesis time can be divided into the learning time needed to build the synthesis model and the inference time required to sample data from the model. 
In terms of the overall synthesis time, PrivBench outperforms all the baselines.
Among them, PrivMRF and DPGAN require extensive matrix computations and even utilize the GPU to accelerate these operations, yet they still require a substantial amount of time to train the data synthesis models.
In contrast, PrivBench ensures high fidelity of synthetic data in terms of data distribution and query runtime performance, while requiring only minimal time for model learning and inference.
This suggests that PrivBench can effectively adapt to frequent benchmark updates and even support real-time, high-fidelity benchmark releases.
}

\myparagraph{Multi-relation synthesis}
As shown in Tables \ref{tab:california} and \ref{tab:job-light}, PrivBench remains highly efficient in synthesizing multi-table databases, with a significantly lower time cost than PrivLava. 
Note that in our experiment, PrivLava also employs the GPU to greatly accelerate its model learning and inference, while PrivBench relies solely on CPU computation.
Although PrivBench requires traversing leaf nodes to construct fanout tables, the workload of fanout construction is light, and the number of leaf nodes is linear to both the number of records and attributes.
Consequently, modeling primary-foreign key dependencies among tables is also efficient in PrivBench.

\section{Related Work}
\label{sec:related}

\subsubsection*{DP-based data synthesis}
\changemarker{With the growing awareness of user privacy and the introduction of data protection regulations worldwide, DP-based data synthesis has gained increasing attention in recent years~\cite{li2014dpsynthesizer, snoke2018pmse, chen2020gs, ge2020kamino, torkzadehmahani2019dp, xie2018differentially, jordon2018pate, liu2023tabular, long2021g, cai2021privmrf, mckenna2019graphical, zhang2017privbayes, cai2023privlava, pujol2022prefair, ping2017datasynthesizer, li2021dpsyn, zhang2021privsyn, mckenna2021winning, vietri2020new, mckenna2022aim}. 
The vast majority of existing work focuses on enhancing the fidelity of the synthetic data in data distribution while overlooking the fidelity in query runtime performance, making them less promising for benchmark publishing scenarios. 
Among them, deep learning-based methods~\cite{chen2020gs, ge2020kamino, torkzadehmahani2019dp, xie2018differentially, jordon2018pate, liu2023tabular, long2021g} train neural networks to fit the joint distribution of the original data. 
Query workload-based methods~\cite{mckenna2021winning, vietri2020new, mckenna2022aim} optimize some given queries' accuracy, rather than their runtime performance, to learn a workload-optimal data distribution. 
Graphical methods~\cite{cai2021privmrf, mckenna2019graphical, zhang2017privbayes, cai2023privlava, pujol2022prefair, ping2017datasynthesizer} achieve high-fidelity data distribution by learning the marginal distribution of the data through graphical models, but these efforts still neglect query runtime performance. 
To the best of our knowledge, PrivBench is the first differentially private data synthesis framework that simultaneously optimizes fidelity in both data distribution and query runtime. 
Additionally, apart from the recent SOTA, PrivLava~\cite{cai2023privlava}, PrivBench is the only framework that can synthesize multi-relation databases while ensuring database-level DP.}

\subsubsection*{SPN-based data management}
\changemarker{SPNs~\cite{poon2011sum}, which excel in representing complex dependencies and data distributions, have been employed for various data analysis tasks, such as image processing and natural language processing~\cite{sanchez2021sum}.
However, the application of SPNs in data management tasks remains underdeveloped.
Hilprecht et al.~\cite{DeepDB} have employed SPNs for approximate query processing and cardinality estimation. 
Recently, Kroes et al.~\cite{kroes2023generating} proposed the first SPN-based data synthesis method. 
However, unlike PrivBench, their approach only supports single-table databases and does not provide privacy guarantees. 
To the best of our knowledge, PrivBench provides the first differentially private method for SPN construction.
Moreover, Treiber~\cite{treiber2020cryptospn} proposed an SPN inference framework based on secure multiparty computation, called CryptoSPN, which ensures that parties with distributed data can compute query results without sharing data.
However, since CryptoSPN does not protect the privacy of query results, attackers may still infer the original data from the query results, which compromises privacy.}

\section{Conclusion}
\label{sec:concl}
In this paper, we delve into the domain of synthesizing databases that preserve privacy for benchmark publishing. Our focus is on creating a database that upholds DP while ensuring that the performance of query workloads on the synthesized data closely aligns with the original data. We propose PrivBench, an innovative synthesis framework designed to generate high-fidelity data while ensuring robust privacy protection. PrivBench utilizes SPNs at its core to segment and sample data from SPN leaf nodes and conducts subsequent operations on these nodes to ensure privacy preservation. It allows users to adjust the granularity of SPN partitions and determine privacy budgets through parameters, crucial for customizing levels of privacy preservation. The data synthesis algorithm is proven to uphold DP. Experimental results highlight PrivBench's capability to create data that not only maintains privacy but also demonstrates high query performance fidelity, showcasing improvements in query runtime discrepancy, query cardinality error, and KLD over alternative approaches.
The promising research directions for future work are twofold: (i) identifying the optimal privacy budget allocation scheme and (ii) generating data and query workloads simultaneously.

\begin{acks}
This work was supported by National Key R\&D program of China 2021YFB3301500, SZU Research Instrument Development and Cultivation 2023YQ017, Guangdong Province Key Laboratory of PHPC 2017B030314073, JSPS KAKENHI 21K19767, 23K17456, 23K24851, 23K25157, 23K28096, and CREST JPMJCR22M2. 
\end{acks}

\clearpage
\balance
\bibliographystyle{abbrv}
\bibliography{validate}

\clearpage

\appendix
\onecolumn

\section{Missing Proofs}
\label{sec:appendix}

\begin{proof}[Proof of Theorem \ref{thm:sumnode}]
Let $\mathcal{M}(T, (\mathcal{S}_1,\dots, \mathcal{S}_K))=\big(\mathcal{M}_1(T[\mathcal{S}_1]),\dots,\mathcal{M}_K(T[\mathcal{S}_K])\big)$. Given a row partition $(\mathcal{S}_1,\dots, \mathcal{S}_K)$, for neighboring tables $T, T'$ and for any possible output $o=(o_1,\dots,o_K)$ of $\mathcal{M}$, we have
\begin{align*}
    &\Pr{\mathcal{M}(T, (\mathcal{S}_1,\dots, \mathcal{S}_K))=o} = \prod_{k\in [K]} \Pr{\mathcal{M}_k(T[\mathcal{S}_k]=o_k)},\\
    &\Pr{\mathcal{M}(T', (\mathcal{S}_1,\dots, \mathcal{S}_K))=o} = \prod_{k\in [K]} \Pr{\mathcal{M}_k(T'[\mathcal{S}_k]=o_k)}.
\end{align*}
Because $(\mathcal{S}_1,\dots, \mathcal{S}_K)$ is a row partition such that $\cap_{k\in[K]}\mathcal{S}_k = \emptyset$, there exists at most one element $\mathcal{S}_j$ of the partition such that
\begin{align*}
    T[\mathcal{S}_j] \sim T'[\mathcal{S}_j] \text{, and } \forall k \neq j, &  T[\mathcal{S}_j] = T'[\mathcal{S}_j].
\end{align*}
Then, because $\mathcal{M}_j$ satisfies table-level $\epsilon_j$-DP, we have
    \begin{align*}
        \forall o_j, \Pr{\mathcal{M}_j(T[\mathcal{S}_j])=o_j} \leq \exp(\epsilon_j) \cdot \Pr{\mathcal{M}_j(T'[\mathcal{S}_j])=o_j}
    \end{align*}
Therefore, given a row partition $(\mathcal{S}_1,\dots, \mathcal{S}_K)$, for any pair of neighboring tables $T, T'$ and for any possible output $o=(o_1,\dots,o_K)$ of $\mathcal{M}$, we have
\begin{align*}
    &\Pr{\mathcal{M}(T, (\mathcal{S}_1,\dots, \mathcal{S}_K))=o} \\
    = & \prod_{k\in [K]} \Pr{\mathcal{M}_k(T[\mathcal{S}_k]=o_k)}\\
    = & \Pr{\mathcal{M}_j(T[\mathcal{S}_j]=o_j)} \cdot \prod_{k\in [K], k\neq j}  \Pr{\mathcal{M}_k(T[\mathcal{S}_k]=o_k)}\\
    \leq & \exp(\epsilon_j) \cdot \Pr{\mathcal{M}_j(T'[\mathcal{S}_j])=o_j} \cdot \prod_{k\in [K], k\neq j}  \Pr{\mathcal{M}_k(T'[\mathcal{S}_k]=o_k)} \\
    = & \exp(\epsilon_j) \cdot \prod_{k\in [K]} \cdot \Pr{\mathcal{M}_k(T'[\mathcal{S}_k]=o_k)} \\
    \leq & \max\set{\epsilon_1,\dots,\epsilon_K} \cdot \Pr{\mathcal{M}_k(T'[\mathcal{S}_k]=o_k)} \\
    = & \max\set{\epsilon_1,\dots,\epsilon_K} \cdot \Pr{\mathcal{M}(T', (\mathcal{S}_1,\dots, \mathcal{S}_K))=o}
\end{align*}
Therefore, we conclude that publishing $\mathcal{M}_1(T[\mathcal{S}_1]),\dots,\mathcal{M}_K(T[\mathcal{S}_K])$ satisfies $\max\set{\epsilon_1,\dots,\epsilon_K}$-DP at the table level.
\end{proof}




\begin{proof}[Proof of Theorem \ref{thm:down-top}]
    We proceed with the following case-by-case analysis.\\
    \textit{Case 1 ($|attr(T)| = 1$):}
    Because $\epsilon = \epsilon_\kw{leaf}$, and $\sigma(T) = 1$, $\kw{PrivSPN}(T, \epsilon)$ satisfies $(\epsilon_\kw{leaf} \cdot \sigma(T))$-DP. \\
    \textit{Case 2 ($|attr(T)| > 1$):}
    Because $\epsilon_\kw{eval} + \epsilon_\kw{op} \leq \epsilon / \sigma(T)$ and $\epsilon_\kw{eval} + \epsilon_\kw{op} + \bar{\epsilon} = \epsilon$, we have
    \begin{align*}
        &\bar{\epsilon} = \epsilon - (\epsilon_\kw{eval} + \epsilon_\kw{op}) \geq \epsilon - \epsilon / \sigma(T)  \Rightarrow \epsilon \leq \frac{\sigma(T)}{\sigma(T)-1}\cdot \bar{\epsilon}. 
    \end{align*}
    When $\kw{op} = \kw{OP}.\kw{SUM}$, given that $\bar{\epsilon} = \epsilon_L$, we have
    \begin{align*}
        \epsilon \leq \frac{\sigma(T)}{\sigma(T)-1}\cdot \bar{\epsilon} = \frac{\sigma(T)}{\sigma(T)-1}\cdot \epsilon_L \leq \frac{\sigma(T)}{\sigma(T[\widetilde{\mathcal{S}}_L])} \cdot \epsilon_L. 
    \end{align*}
    When $\kw{op} = \kw{OP}.\kw{PRODUCT}$, given that $\bar{\epsilon} = \frac{\sigma(T[\widetilde{\mathcal{S}}_L])+\sigma(T[\widetilde{\mathcal{S}}_R])}{\sigma(T[\widetilde{\mathcal{S}}_L])} \cdot \epsilon_L$, we have
    \begin{align*}
        \epsilon & \leq \frac{\sigma(T)}{\sigma(T)-1}\cdot \bar{\epsilon} = \frac{\sigma(T)}{\sigma(T)-1}\cdot \frac{\sigma(T[\widetilde{\mathcal{S}}_L])+\sigma(T[\widetilde{\mathcal{S}}_R])}{\sigma(T[\widetilde{\mathcal{S}}_L])} \cdot \epsilon_L \\
        & = \frac{\sigma(T)}{\sigma(T[\widetilde{\mathcal{S}}_L])} \cdot \frac{\sigma(T[\widetilde{\mathcal{S}}_L])+\sigma(T[\widetilde{\mathcal{S}}_R])}{\sigma(T)-1} \cdot \epsilon_L \\
        & = \frac{\sigma(T)}{\sigma(T[\widetilde{\mathcal{S}}_L])} \cdot \frac{2|T[\widetilde{\mathcal{S}}_L]|\cdot |attr(T[\widetilde{\mathcal{S}}_L])| / \beta - 1 + 2|T[\widetilde{\mathcal{S}}_R]|\cdot |attr(T[\widetilde{\mathcal{S}}_R])| / \beta - 1}{2|T|\cdot |attr(T)| / \beta - 2} \cdot \epsilon_L \\
        & = \frac{\sigma(T)}{\sigma(T[\widetilde{\mathcal{S}}_L])} \cdot \frac{2|T[\widetilde{\mathcal{S}}_L]|\cdot |attr(T[\widetilde{\mathcal{S}}_L])|  + 2|T[\widetilde{\mathcal{S}}_R]|\cdot |attr(T[\widetilde{\mathcal{S}}_R])|  - 2\beta}{2|T|\cdot |attr(T)| - 2\beta} \cdot \epsilon_L \\
        & = \frac{\sigma(T)}{\sigma(T[\widetilde{\mathcal{S}}_L])} \cdot \frac{2|T|\cdot |attr(T[\widetilde{\mathcal{S}}_L])|  + 2|T|\cdot |attr(T[\widetilde{\mathcal{S}}_R])|  - 2\beta}{2|T|\cdot |attr(T)| - 2\beta} \cdot \epsilon_L \\
        & = \frac{\sigma(T)}{\sigma(T[\widetilde{\mathcal{S}}_L])} \cdot \frac{2|T|\cdot (|attr(T[\widetilde{\mathcal{S}}_L])|  + |attr(T[\widetilde{\mathcal{S}}_R])|)  - 2\beta}{2|T|\cdot |attr(T)| - 2\beta} \cdot \epsilon_L \\
            & = \frac{\sigma(T)}{\sigma(T[\widetilde{\mathcal{S}}_L])} \cdot \frac{2|T|\cdot |attr(T)|  - 2\beta}{2|T|\cdot |attr(T)| - 2\beta} \cdot \epsilon_L = \frac{\sigma(T)}{\sigma(T[\widetilde{\mathcal{S}}_L])} \cdot \epsilon_L.
    \end{align*}
    Therefore, for any table with at least two attributes, regardless of whether $\kw{op} = \kw{OP}.\kw{SUM}$ or $\kw{op} = \kw{OP}.\kw{PRODUCT}$, we have $\epsilon \leq \frac{\sigma(T)}{\sigma(T[\widetilde{\mathcal{S}}_L])} \cdot \epsilon_L$.
    Assume that the leftmost path of the binary tree returned by $\kw{PrivSPN}(T, \epsilon)$ contains $Q$ non-leaf nodes.
    For all $q \in [1, Q]$, let $S^{(q)}$ denote the left subtable produced by the non-leaf node at depth $(q-1)$ along this path.
    Consequently, we can obtain the following recursively: 
    \begin{align*}
        \epsilon \leq \frac{\sigma(T)}{\sigma(T[\widetilde{\mathcal{S}}_L])} \cdot \epsilon_L = \frac{\sigma(T)}{\sigma(S^{(1)})} \cdot \epsilon_L \leq \frac{\sigma(T)}{\sigma(S^{(1)})} \cdot \frac{\sigma(S^{(1)})}{\sigma(S^{(2)})} \cdot \dots \cdot \frac{\sigma(S^{(Q-1)})}{\sigma(S^{(Q)})} \cdot \epsilon_\kw{leaf} = \frac{\sigma(T)}{\sigma(S^{(Q)})} \cdot \epsilon_\kw{leaf} = \sigma(T) \cdot \epsilon_\kw{leaf}.
    \end{align*}
    Therefore, we conclude that $\kw{PrivSPN}(T, \epsilon)$ satisfies table-level $(\sigma(T) \cdot \epsilon_\kw{leaf})$-DP.
\end{proof}



\begin{proof}[Proof of Theorem \ref{thm:table-to-db-dp}]
    Because each $\mathcal{M}_i$ satisfies $\epsilon_i$-DP at the table level, for any pair of neighboring databases $D, D'$ and for any possible output $o=(o_1,\dots,o_n)$ of $\mathcal{M}$, we have
    \begin{align*}
        & \Pr{\mathcal{M}(D)=o} = \prod_{i\in[n]} \Pr{\mathcal{M}_i(T_i)=o_i} \leq  \prod_{i\in[n]}  \exp(d_i\cdot \epsilon_i)\cdot \Pr{\mathcal{M}_i(T_i^{(d_i)})=o_i} =  \exp(\sum_{i\in[n]} d_i \cdot \epsilon_i) \cdot \Pr{\mathcal{M}(D')=o}
    \end{align*}
    where $T_i^{(d_i)}$ is a private table in $D'$ that differs from $T_i$ in the values of $d_i$ rows for all $i\in[n]$.
    Then, because $T_i, T_i^{(d_i)}$ differ in at most $\tau_i$ rows for all $i\in[n]$, we have
    \begin{align*}
        \Pr{\mathcal{M}(D)=o} \leq  \exp(\sum_{i\in[n]} d_i \cdot \epsilon_i) \cdot \Pr{\mathcal{M}(D')=o} \leq  \exp(\sum_{i\in[n]} \tau_i \cdot \epsilon_i) \cdot \Pr{\mathcal{M}(D')=o}
    \end{align*}
    Therefore, $\mathcal{M}(D)$ satisfies $(\sum_{i\in[n]}\tau_i\cdot\epsilon_i)$-DP at the database level.
\end{proof}


\begin{proof}[Proof of Corollary \ref{coro:privbench}]
    For any pair of neighboring databases $D, D'$, for any possible output $t_i$ of \kw{PrivSPN}, for any possible output $t_i'$ of \kw{PrivFanout}, and for any possible output $\widehat{D}=\set{\widehat{T}_1,...,\widehat{T}_n}$ of \kw{PrivBench}, we have
    \begin{align*}
        &\Pr{\kw{PrivBench}(D)=\widehat{D}}\\
    = & \prod_{i\in[n]}\Pr{\kw{PrivSPN}(T_i, \epsilon^s_i) = t_i} \cdot \prod_{T_i \text{ refers to } T_j} \Pr{\kw{PrivFanout}(T_i, t_i, FK_{i,j}, \epsilon^f_i) = t_i'} \cdot \prod_{i\in[n]}\Pr{\kw{SampleDataFromSPN}(t_i')=\widehat{D}} \\
    \leq & \prod_{i\in[n]} \exp(d_i\cdot \epsilon^s_i) \cdot \Pr{\kw{PrivSPN}(T_i^{(d_i)}, \epsilon^s_i)=t_i}  \cdot \prod_{T_i \text{ refers to } T_j} \exp(d_i\cdot \epsilon^f_i) \cdot \Pr{\kw{PrivFanout}(T_i^{(d_i)}, t_i, FK_{i,j}, \epsilon^f_i)=t_i'} \\
    & \cdot \prod_{i\in[n]}\Pr{\kw{SampleDataFromSPN}(t_i')=\widehat{D}} \\
    = & \exp(\sum_{i\in[n]}d_i\cdot \epsilon^s_i + \sum_{T_i \text{ refers to } T_j} d_i\cdot \epsilon^f_i)  \cdot \prod_{i\in[n]} \Pr{\kw{PrivSPN}(T_i^{(d_i)}, \epsilon^s_i)=t_i}  \cdot \prod_{T_i \text{ refers to } T_j} \Pr{\kw{PrivFanout}(T_i^{(d_i)}, t_i, FK_{i,j}, \epsilon^f_i)=t_i'} \\
    & \cdot \prod_{i\in[n]}\Pr{\kw{SampleDataFromSPN}(t_i')=\widehat{D}}\\
    = & \exp(\sum_{i\in[n]} d_i\cdot \epsilon^s_i + \sum_{T_i \text{ refers to } T_j} d_i\cdot \epsilon^f_i)  \cdot \Pr{\kw{PrivBench}(D')=\widehat{D}} \\
    \leq & \exp(\sum_{i\in[n]}\tau_i\cdot \epsilon^s_i + \sum_{T_i \text{ refers to } T_j} \tau_i\cdot \epsilon^f_i) \cdot \Pr{\kw{PrivBench}(D')=\widehat{D}}
    \end{align*}
    where $T_i^{(d_i)}$ is a private table in $D'$ that differs from $T_i$ in the values of $d_i$ rows for all $i\in[n]$. 
    Therefore, \kw{PrivBench} satisfies $(\sum_{i\in[n]}\tau_i\cdot\epsilon^s_i+\sum_{T_i \text{ refers to } T_j} \tau_i\cdot\epsilon^f_i)$-DP at the database level.
\end{proof}

\end{document}


\title{Online Appendix: Privacy-Enhanced Database Synthesis for Benchmark Publishing}


\author{Yunqing Ge}
\affiliation{%
  \institution{\normalsize{Shenzhen University}}
}
\email{geyunqing2022@email.szu.edu.cn}

\author{Jianbin Qin}
\authornote{Corresponding authors.}
\affiliation{%
  \institution{\normalsize{SICS, Shenzhen University}}
}
\email{qinjianbin@szu.edu.cn}

\author{Shuyuan Zheng}
\authornotemark[1]
\affiliation{%
  \institution{\normalsize{Osaka Univeristy}}
}
\email{zheng@ist.osaka-u.ac.jp}

\author{Yongrui Zhong}
\affiliation{%
  \institution{\normalsize{Shenzhen University}}
}
\email{zhongyongrui2021@email.szu.edu.cn}

\author{Bo Tang}
\affiliation{%
  \institution{\normalsize{Southern University of Science and Technology}}
}
\email{tangb3@sustech.edu.cn}

\author{Yu-Xuan Qiu}
\affiliation{%
  \institution{\normalsize{Beijing Institute of Technology}}
}
\email{qiuyx.cs@gmail.com}

\author{Rui Mao}
\affiliation{%
  \institution{\normalsize{SICS, Shenzhen University}}
}
\email{mao@szu.edu.cn}

\author{Ye Yuan}
\affiliation{%
  \institution{\normalsize{Beijing Institute of Technology}}
}
\email{yuan-ye@bit.edu.cn}

\author{Makoto Onizuka}
\affiliation{%
  \institution{\normalsize{Osaka Univeristy}}
}
\email{onizuka@ist.osaka-u.ac.jp}

\author{Chuan Xiao}
\affiliation{%
  \institution{\normalsize{Osaka Univeristy, Nagoya University}}
}
\email{chuanx@ist.osaka-u.ac.jp}

\maketitle


\section{Missing Proofs}


   



\begin{theorem}[Parallel composition under bounded DP]
\label{thm:sumnode}
    Given a row partition $(\mathcal{S}_1,\dots, \mathcal{S}_K)$, publishing $\mathcal{M}_1(T[\mathcal{S}_1]),\dots,\mathcal{M}_K(T[\mathcal{S}_K])$ satisfies $\max\set{\epsilon_1,\dots, \epsilon_K}$-DP at the database level, where $\mathcal{S}_k$ is a subset of row indices and $\mathcal{M}_k$ is a table-level $\epsilon_k$-DP algorithm, $\forall k \in [K]$.
\end{theorem}

\begin{proof}
Let $\mathcal{M}(T, (\mathcal{S}_1,\dots, \mathcal{S}_K))=\big(\mathcal{M}_1(T[\mathcal{S}_1]),\dots,\mathcal{M}_K(T[\mathcal{S}_K])\big)$. Given a row partition $(\mathcal{S}_1,\dots, \mathcal{S}_K)$, for neighboring tables $T, T'$ and for any possible output $o=(o_1,\dots,o_K)$ of $\mathcal{M}$, we have
\begin{align*}
    &\Pr{\mathcal{M}(T, (\mathcal{S}_1,\dots, \mathcal{S}_K))=o} = \prod_{k\in [K]} \Pr{\mathcal{M}_k(T[\mathcal{S}_k]=o_k)},\\
    &\Pr{\mathcal{M}(T', (\mathcal{S}_1,\dots, \mathcal{S}_K))=o} = \prod_{k\in [K]} \Pr{\mathcal{M}_k(T'[\mathcal{S}_k]=o_k)}.
\end{align*}
Because $(\mathcal{S}_1,\dots, \mathcal{S}_K)$ is a row partition such that $\cap_{k\in[K]}\mathcal{S}_k = \emptyset$, there exists at most one element $\mathcal{S}_j$ of the partition such that
\begin{align*}
    T[\mathcal{S}_j] \sim T'[\mathcal{S}_j] \text{, and } \forall k \neq j, &  T[\mathcal{S}_j] = T'[\mathcal{S}_j].
\end{align*}
Then, because $\mathcal{M}_j$ satisfies table-level $\epsilon_j$-DP, we have
    \begin{align*}
        \forall o_j, \Pr{\mathcal{M}_j(T[\mathcal{S}_j])=o_j} \leq \exp(\epsilon_j) \cdot \Pr{\mathcal{M}_j(T'[\mathcal{S}_j])=o_j}
    \end{align*}
Therefore, given a row partition $(\mathcal{S}_1,\dots, \mathcal{S}_K)$, for any pair of neighboring tables $T, T'$ and for any possible output $o=(o_1,\dots,o_K)$ of $\mathcal{M}$, we have
\begin{align*}
    &\Pr{\mathcal{M}(T, (\mathcal{S}_1,\dots, \mathcal{S}_K))=o} \\
    = & \prod_{k\in [K]} \Pr{\mathcal{M}_k(T[\mathcal{S}_k]=o_k)}\\
    = & \Pr{\mathcal{M}_j(T[\mathcal{S}_j]=o_j)} \cdot \prod_{k\in [K], k\neq j}  \Pr{\mathcal{M}_k(T[\mathcal{S}_k]=o_k)}\\
    \leq & \exp(\epsilon_j) \cdot \Pr{\mathcal{M}_j(T'[\mathcal{S}_j])=o_j} \cdot \prod_{k\in [K], k\neq j}  \Pr{\mathcal{M}_k(T'[\mathcal{S}_k]=o_k)} \\
    = & \exp(\epsilon_j) \cdot \prod_{k\in [K]} \cdot \Pr{\mathcal{M}_k(T'[\mathcal{S}_k]=o_k)} \\
    \leq & \max\set{\epsilon_1,\dots,\epsilon_K} \cdot \Pr{\mathcal{M}_k(T'[\mathcal{S}_k]=o_k)} \\
    = & \max\set{\epsilon_1,\dots,\epsilon_K} \cdot \Pr{\mathcal{M}(T', (\mathcal{S}_1,\dots, \mathcal{S}_K))=o}
\end{align*}
Therefore, we conclude that publishing $\mathcal{M}_1(T[\mathcal{S}_1]),\dots,\mathcal{M}_K(T[\mathcal{S}_K]$ satisfies $\max\set{\epsilon_1,\dots,\epsilon_K}$-DP at the table level.
\end{proof}



\begin{theorem}
  \label{thm:down-top}
    Let $\mathcal{T}$ denote the set of all possibles subtables of $T$ such that $|\mathcal{T}_k| < 2\beta$ and $|\mathcal{T}_k.\kw{Attr}|=1$, $\forall \mathcal{T}_k \in \mathcal{T}$.
    If for any $\mathcal{T}_k\in \mathcal{T}$, $\kw{PrivSPN}(T, \epsilon)$ calls $\kw{PrivSPN}(\mathcal{T}_k, \epsilon_\kw{leaf})$, $\kw{PrivSPN}(T, \epsilon)$ satisfies table-level $(\sigma(T) \cdot \epsilon_\kw{leaf})$-DP, where $\sigma(T)=2^{|T.\kw{Attr}| + \frac{|T|}{\beta} - 2}\cdot |T.\kw{Attr}| \cdot \frac{|T|}{\beta}$.  
\end{theorem}

\begin{proof}
    For any table $T$, the  $\kw{PrivSPN}(T, \epsilon)$ procedure that sequentially calls $\kw{Planning}(T, \epsilon)$, $\kw{ParentGen}(T, \kw{op}, \epsilon_\kw{op})$, and $\kw{ChildrenGen}(T, \kw{op}, \widetilde{\mathcal{S}}_L, \widetilde{\mathcal{S}}_R, \bar{\epsilon})$ satisfies $(\gamma_1\cdot \epsilon + \epsilon_\kw{op} + \bar{\epsilon})$-DP. 
    Because $\bar{\epsilon} \geq \frac{1}{2} \epsilon$ and $\gamma_1\cdot \epsilon + \epsilon_\kw{op} + \bar{\epsilon} = \epsilon$, $\kw{PrivSPN}(T, \epsilon)$ satisfies $(2\bar{\epsilon})$-DP.
    When $\kw{op} \neq \kw{OP.LEAF}$, because $\kw{ChildrenGen}(T, \kw{op}, \widetilde{\mathcal{S}}_L, \widetilde{\mathcal{S}}_R, \bar{\epsilon})$ calls $\kw{PrivSPN}(T[\widetilde{\mathcal{S}}_L], \epsilon_{l})$ and $\kw{PrivSPN}(T[\widetilde{\mathcal{S}}_R], \epsilon_r)$, where $\epsilon_l + \epsilon_r \geq \bar{\epsilon}$, $\kw{PrivSPN}(T, \epsilon)$ satisfies $(2\epsilon_l + 2\epsilon_r)$-DP.
    Similarly, by recursively repeating this reasoning process, we can know that $\kw{PrivSPN}(T, \epsilon)$ satisfies $(\sum_{k\in[K]} 2^{d_k}\cdot \epsilon_\kw{leaf})$-DP, where $K$ is the number of leaf nodes of $\kw{PrivSPN}(T, \epsilon)$ and $d_k$ is the depth (starting from $0$) of the $k$-th leaf node.

    Then, because we can apply at most $(|T|/\beta - 1)$ times of row splitting and at most $(|T.\kw{Attr}| - 1)$ times of column splitting to obtain any subtable $\mathcal{T}_k \in \mathcal{T}$ for a leaf node, a leaf node can at most have $(|T|/\beta - 1)$ sum nodes and at most $(|T.\kw{Attr}| - 1$ product nodes as its ancestor nodes.
    Therefore, we have
    \begin{align*}
        \forall k \in [K], d_k \leq |T.\kw{Attr}| + \frac{|T|}{\beta} - 2
    \end{align*}
    
    Additionally, because each subtable $\mathcal{T}_k$ of $T$ for a leaf node satisfies $|\mathcal{T}_k| \geq \beta$ and $|\mathcal{T}_k.\kw{Attr}| = 1$, $\kw{PrivSPN}(T, \epsilon)$ can have at most $|T.\kw{Attr}| \cdot \frac{|T|}{\beta}$ leaf nodes, i.e., $K \leq |T.\kw{Attr}| \cdot \frac{|T|}{\beta}$.
    Therefore, we have 
    \begin{align*}
        & \sum_{k\in[K]} 2^{d_k}\cdot \epsilon_\kw{leaf}  \\
        \leq & \sum_{k\in[K]} 2^{|T.\kw{Attr}| + \frac{|T|}{\beta} - 2}\cdot \epsilon_\kw{leaf} =  K \cdot 2^{|T.\kw{Attr}| + \frac{|T|}{\beta} - 2}\cdot \epsilon_\kw{leaf} \\
        \leq & |T.\kw{Attr}| \cdot \frac{|T|}{\beta} \cdot 2^{|T.\kw{Attr}| + \frac{|T|}{\beta} - 2}\cdot \epsilon_\kw{leaf} = \sigma(T) \cdot \epsilon_\kw{leaf}
    \end{align*}
    Therefore, we conclude that $\kw{PrivSPN}(T, \epsilon)$ satisfies table-level $(\sigma(T) \cdot \epsilon_\kw{leaf})$-DP.
\end{proof}


\begin{theorem}
\label{thm:table-to-db-dp}
    Given $\mathcal{M}_i$ that satisfies $\epsilon_i$-DP at the table level for all $i\in[n]$, $\mathcal{M}(D)=\big(\mathcal{M}_1(T_1),...,\mathcal{M}_n(T_n)\big)$ satisfies $\sum_{i\in[n]}\tau_i\cdot\epsilon_i$-DP at the database level, where $\tau_i$ is the maximum multiplicity of $T_i$.
\end{theorem}

\begin{proof}
    Because each $\mathcal{M}_i$ satisfies $\epsilon_i$-DP at the table level, for any pair of neighboring databases $D, D'$ and for any possible output $o=(o_1,\dots,o_n)$ of $\mathcal{M}$, we have
    \begin{align*}
        & \Pr{\mathcal{M}(D)=o}\\
        = & \prod_{i\in[n]} \Pr{\mathcal{M}_i(T_i)=o_i} \\
        \leq & \prod_{i\in[n]}  \exp(k_i\cdot \epsilon_i)\cdot \Pr{\mathcal{M}_i(T_i^{(k_i)})=o_i}\\
        = & \exp(\sum_{i\in[n]} k_i \cdot \epsilon_i) \cdot \Pr{\mathcal{M}(D')=o}
    \end{align*}
    where $T_i^{(k_i)}$ is a private table in $D'$ that differs from $T_i$ in the values of $k_i$ rows for all $i\in[n]$.
    Then, because $T_i, T_i^{(k_i)}$ differ in at most $\tau_i$ rows for all $i\in[n]$, we have
    \begin{align*}
        & \Pr{\mathcal{M}(D)=o}\\
        \leq & \exp(\sum_{i\in[n]} k_i \cdot \epsilon_i) \cdot \Pr{\mathcal{M}(D')=o} \\
        \leq & \exp(\sum_{i\in[n]} \tau_i \cdot \epsilon_i) \cdot \Pr{\mathcal{M}(D')=o}
    \end{align*}
    Therefore, $\mathcal{M}(D)$ satisfies $\sum_{i\in[n]}\tau_i\cdot\epsilon_i$-DP at the database level.
\end{proof}

\begin{corollary}[of Theorem \ref{thm:table-to-db-dp}]
\label{coro:privbench}
    \kw{PrivBench} satisfies $(\sum_{i\in[n]}\tau_i\cdot\epsilon^s_i+\sum_{T_i \text{ refers to } T_j} \tau_i\cdot\epsilon^f_i)$-DP at the database level.
\end{corollary}

\begin{proof}
    For any pair of neighboring databases $D, D'$, for any possible output $t_i$ of \kw{PrivSPN}, for any possible output $t_i'$ of \kw{PrivFanout}, and for any possible output $\widehat{D}=\set{\widehat{T}_1,...,\widehat{T}_n}$ of \kw{PrivBench}, we have
    \begin{align*}
        &\Pr{\kw{PrivBench}(D)=\widehat{D}}\\
    = & \prod_{i\in[n]}\Pr{\kw{PrivSPN}(T_i, \epsilon^s_i) = t_i} \cdot \prod_{T_i \text{ refers to } T_j} \Pr{\kw{PrivFanout}(T_i, t_i, PK_j, \epsilon^f_i) = t_i'}  \cdot \prod_{i\in[n]}\Pr{\kw{SampleDataFromSPN}(t_i')=\widehat{D}} \\
    \leq & \prod_{i\in[n]} \exp(k_i\cdot \epsilon^s_i) \cdot \Pr{\kw{PrivSPN}(T_i^{(k_i)}, \epsilon^s_i)=t_i} \cdot \prod_{T_i \text{ refers to } T_j} \exp(k_i\cdot \epsilon^f_i) \cdot \Pr{\kw{PrivFanout}(T_i^{(k_i)}, t_i, PK_j, \epsilon^f_i)=t_i'} \\
    & \cdot \prod_{i\in[n]}\Pr{\kw{SampleDataFromSPN}(t_i')=\widehat{D}} \\
    = & \exp(\sum_{i\in[n]}k_i\cdot \epsilon^s_i + \sum_{T_i \text{ refers to } T_j} k_i\cdot \epsilon^f_i) \cdot \prod_{i\in[n]} \Pr{\kw{PrivSPN}(T_i^{(k_i)}, \epsilon^s_i)=t_i} \cdot \prod_{T_i \text{ refers to } T_j} \Pr{\kw{PrivFanout}(T_i^{(k_i)}, t_i, PK_j, \epsilon^f_i)=t_i'} \\
    & \cdot \prod_{i\in[n]}\Pr{\kw{SampleDataFromSPN}(t_i')=\widehat{D}}\\
    = & \exp(\sum_{i\in[n]} k_i\cdot \epsilon^s_i + \sum_{T_i \text{ refers to } T_j} k_i\cdot \epsilon^f_i) \cdot \Pr{\kw{PrivBench}(D')=\widehat{D}} \leq  \exp(\sum_{i\in[n]}\tau_i\cdot \epsilon^s_i + \sum_{T_i \text{ refers to } T_j} \tau_i\cdot \epsilon^f_i) \cdot \Pr{\kw{PrivBench}(D')=\widehat{D}}
    \end{align*}
    where $T_i^{(k_i)}$ is a private table in $D'$ that differs from $T_i$ in the values of $k_i$ rows for all $i\in[n]$. 
    Therefore, \kw{PrivBench} satisfies $(\sum_{i\in[n]}\tau_i\cdot\epsilon^s_i+\sum_{T_i \text{ refers to } T_j} \tau_i\cdot\epsilon^f_i)$-DP at the database level.
\end{proof}

\section{Auxiliary Information}
Table \ref{tab:datasets} summarizes the datasets, workloads, and baselines used in our experiments, and Table \ref{tab:notation} lists all important notations used throughout the paper.

\begin{table*}[ht]
\caption{Summary of datasets, workloads, and baselines.}
\label{tab:datasets}
\begin{tabular}{@{}cccccc@{}}
\toprule
\textbf{Dataset} & \textbf{\#Tables} & \textbf{\#Rows} & \textbf{\#Attributes} & \textbf{Workload} & \textbf{Baselines} \\ \midrule
Adult & 1 & 48842 & 13 & SAM-1000 & \begin{tabular}[c]{@{}c@{}}SAM (non-DP), VGAN (non-DP),\\ AIM, DPSynthesizer, CGAN, PrivSyn, \\ DataSyn, PreFair, MST, PrivMRF\end{tabular} \\ \midrule
California & 2 & 2.3M & 28 & MSCN-400 & SAM (non-DP), PrivLava \\ \midrule
JOB-light & 6 & 57M & 15 & \begin{tabular}[c]{@{}c@{}}JOB-70, \\ MSCN-400\end{tabular} & SAM (non-DP), DPSynthesizer \\ \bottomrule
\end{tabular}
\end{table*}

\begin{table}[ht]
\small
\caption{Important notations.}
\label{tab:notation}
\begin{tabular}{@{}cc@{}}
\toprule
Notation & Meaning \\ \midrule
$D$ & Original database \\
$\hat{D}$ & Synthetic database \\
$T$ & Original table \\
$\hat{T}$ & Synthetic table \\
$r_i$ & A record in table $T_i$ \\
$D, D'$ & Neighboring databases \\
$\mathcal{S}$ & Subset of row (or column) indices \\
$t/t'$ & An SPN tree with/without fanout \\
$\epsilon$ & Total privacy budget \\
$\kw{op}$ & Operation deciding the type of the parent node \\
$\epsilon_\kw{op}, \epsilon_L, \epsilon_R$ & Privacy budgets for parent, left child, and right child \\
$\epsilon_\kw{op}$ & Privacy budget for the parent node \\
$\bar{\epsilon}$ & Remaining privacy budget \\
$\tilde{S}$ & Perturbed subset of row (or column) indices \\
$H$ & Histogram of a table \\
$\epsilon^s_i$ & Privacy budget for SPN construction for table $i$ \\
$\epsilon^f_i$ & Privacy budget for fanout construction for table $i$ \\
$\epsilon_\kw{eval}$ & Privacy budget for correlation trial \\
$\alpha$ & threshold for column splitting \\
$\beta$ & minimum table size \\
$\gamma_1$ & budget ratio for correlation trial \\
$\gamma_2$ & budget ratio for column splitting in correlation trial \\
$\delta_\kw{col}$ & Correlation between vertically split tables \\
$\Delta$ & Global sensitivity \\
$PK_i$ & Primary key of table $T_i$ \\
$FK_{i,j}$ & Foreign key of table $T_i$ that refers to table $T_j$ \\
$\mathcal{F}_{i, j}$ & Fanout table of table $T_i$ that refers to $T_j$ \\ \bottomrule
\end{tabular}
\end{table}

\clearpage